%% file: main.tex
\documentclass[11pt, letter]{article} 
\usepackage{geometry}
\usepackage{latexsym}
\usepackage{xspace}
\usepackage{boxedminipage}
\usepackage{amsmath}
\usepackage{amsfonts,bm}
\usepackage{amssymb}
\usepackage{tabls}
\usepackage{theorem}
\usepackage[dvips]{graphicx}
\usepackage{color}
\usepackage{enumerate}
\usepackage{algorithm}
\usepackage{algorithmicx}
\usepackage[noend]{algpseudocode}
\newcommand{\ignore}[1]{}
\usepackage{sansmath}
\usepackage{hyperref}
 \usepackage{enumitem}

\newlength{\leftbarwidth}
\newlength{\leftbarsep}
\newtheorem{theorem}{Theorem}[section]
\newtheorem{lemma}[theorem]{Lemma}

\newtheorem{definition}[theorem]{Definition}
\newtheorem{claim}[theorem]{Claim}

\newtheorem{remark}[theorem]{Remark}

\newenvironment{proof}{\noindent{\bf Proof}:}{$\hfill \Box$\\}

\newenvironment{proofof}[1]{\noindent{\bf Proof of #1}:}{$\hfill \Box$\\}

\def\P{\mathcal{P}}
\newcommand {\supp}[1]{\textup{supp}(#1)}
\def\I{\mathcal{I}}
\def\M{\mathcal{M}}

\def\Z{\mathcal{Z}}
\def\U{\altmathcal{U}}

\def\bA{\mathbf{A}}
\def\bC{\mathbf{C}}
\def\bV{\mathbf{V}}
\def\bW{\mathbf{W}}
\def\ba{\mathbf{a}}
\def\bb{\mathbf{b}}
\def\bd{\mathbf{d}}
\def\be{\mathbf{e}}
\def\bx{\mathbf{x}}
\def\bX{\mathbf{X}}
\def\by{\mathbf{y}}
\def\bz{\mathbf{z}}
\def\bp{\mathbf{p}}
\def\bw{\mathbf{w}}
\def\bv{\mathbf{v}}
\def\btheta{\bm{\theta}}
\def\blambda{\bm{\lambda}}
\def\rR{\mathbb{R}}
\def\zZ{\mathbb{Z}}

\newcommand{\EE}{\mathcal{E}}
\renewcommand{\Pr}{\mathbb{P}}
\newcommand{\B}{\mathcal{B}}
\newcommand{\V}{\mathcal{V}}
\newcommand{\C}{\mathcal{C}}

\def\SS{\altmathcal{S}}

\def\etal{et al.}

\def\para#1{\noindent {\bf #1}}

\newcommand{\cut}[1]{}
\include{tao-command}

\newcommand{\poly}{\mathrm{poly}}

\DeclareMathOperator*{\argmin}{\arg\min}
\DeclareMathOperator*{\argmax}{\arg\max}

\makeatletter
\newcommand*{\rom}[1]{\expandafter\@slowromancap\romannumeral #1@}
\makeatother
\title{Maximizing Determinants under Matroid Constraints}

\author{Vivek Madan\thanks{Amazon. This work was done while the author was at Georgia Institute of Technology. Email: vmadan7@gatech.edu}
\and
Aleksandar Nikolov\thanks{University of Toronto. Email:
  anikolov@cs.toronto.edu. Supported by NSERC Discovery Grant.
}
\and
        Mohit Singh\thanks{Georgia Institute of Technology. Email: mohit.singh@isye.gatech.edu. Supported by NSF- AF:1910423 and NSF-AF:1717947.
}
\and
Uthaipon Tantipongpipat\thanks{Georgia Institute of Technology.
Email: tao@gatech.edu. Supported by NSF- AF:1910423 and NSF-AF:1717947.
}
}

\begin{document}

\maketitle

\thispagestyle{empty}
\input{abs}
\newpage


\pagenumbering{arabic}
\input{intro}

\input{solvability}

\input{sparsity}

\input{rounding}

\input{deterministic}

\bibliographystyle{alpha}
\bibliography{references}

\appendix
\input{preliminaries}

\input{log-concavity}

\input{app-relax}
\input{reduction}

\input{app-solve}

\input{app-KKT}

\input{app-sparsity}

\input{app-round}

\input{no-pc-dist}

\input{partition-approx}

\end{document}

%% file: tao-command.tex
\newcommand{\E}{\mathbb{E}}

\newcommand{\X}{\mathcal{X}}
\newcommand{\D}{\mathcal{D}}
\newcommand{\F}{\mathcal{F}}
\newcommand{\G}{\mathcal{G}}

\newcommand{\R}{\mathbb{R}}

\renewcommand{\SS}{\mathbf{S}}
\renewcommand{\U}{\mathcal{U}}
\renewcommand{\P}{\mathcal{P}}

\newcommand{\balpha}{\boldsymbol\alpha}
\newcommand{\bal}{\bm{\alpha}} 
\newcommand{\bY}{\mathbf{Y}}
\newcommand{\bM}{\mathbf{M}}

\newcommand{\LP}[1][x]{\textup{LP}$_\textup{#1-OPT}$}

\newcommand{\eps}{\epsilon}

\usepackage{bbm}
\usepackage{dsfont}
\newcommand{\one}{\mathds{1}}

\newcommand{\pr}[1]{\left(#1\right)}

\newcommand{\set}[1]{\ensuremath{\left\{#1\right\}}}
\newcommand{\an}[1]{\left\langle#1\right\rangle}
\newcommand{\abs}[1]{\left|#1\right|}
\newcommand{\tr}{\mathrm{tr}}

\newcommand{\CPD}{\textup{CP}}
\newcommand{\CP}{\CPOPT}
\newcommand{\CPOPT}{\ensuremath{\textup{OPT}_{\textup{CP}}}}
\newcommand{\Cs}{C^*}
\newcommand{\DOPT}{\textup{D-OPT}}
\newcommand{\OPT}{\textup{OPT}}
\newcommand{\DetMax}{\textsc{Determinant Maximization}}
\newcommand{\DFea}{\textsc{Det-Feasibility}}

%% file: abs.tex
\begin{abstract}
Given a set of vectors $\bv_1,\dots,\bv_n\in \rR^d$ and a matroid $\M = ([n],\I)$, we study the problem of finding a basis $S$ of $\M$ such that $\det\pr{\sum_{i \in S} \bv_i \bv_i^\top}$ is maximized. This problem appears in a diverse set of areas, such as experimental design, fair allocation of goods, network design, and machine learning.  The current best results include an $e^{2k}$-estimation for any matroid of rank $k$~\cite{AnariGV18} and a $(1+\epsilon)^d$-approximation for a uniform matroid of rank $k \geq d + \frac{d}{\epsilon}$~\cite{MadanSTU19}, where the rank $k \ge d$ denotes the desired size of the optimal set. Our main result is a new approximation algorithm for the general problem with an approximation guarantee that depends only on the dimension $d$ of the vectors, and not on the size $k$ of the output set. In particular, we show an $\pr{O(d)}^{d}$-estimation and an $\pr{O(d)}^{d^3}$-approximation for any matroid, giving a significant improvement over prior work when $k \gg d$.

Our result relies on showing that there exists an optimal solution to a convex programming relaxation for the problem which has \emph{sparse support}; in particular, no more than $O(d^2)$ variables of the solution have fractional values. The sparsity results rely on the interplay between the first order optimality conditions for the convex program and matroid theory. We believe that the techniques introduced to show sparsity of optimal solutions to convex programs will be of independent interest. We also give a randomized rounding algorithm that, given a sparse fractional solution to the convex program, returns a feasible integral solution to the original problem. To show the approximation guarantee, we utilize recent works on strongly log-concave polynomials~\cite{AnariGV18,AnariLGV19} and show new relationships between different convex programs~\cite{NikolovS16,anari2017generalization} studied for the problem. We remark that sparsity is crucial to the algorithm and that all previous approaches will necessarily fail to achieve such an improved guarantee. Finally, we show how to use the estimation algorithm to give an efficient deterministic approximation algorithm. Once again, the algorithm crucially relies on sparsity of the fractional solution to guarantee that the approximation factor depends solely on the dimension $d$.
\end{abstract}

%% file: intro.tex
\section{Introduction}

Choosing a diverse representative set of items from a large corpus is a common problem studied in a variety of areas, including machine learning, information retrieval, statistics, and optimization~\cite{kulesza2012determinantal,chen2006less,CevallosEZ17,pukelsheim2006optimal}.
For example, consider the problem of choosing a subset from a large data set to train a machine learning algorithm; or of displaying a small set of images out of a large set of relevant images to a search query. In these contexts, one aims to choose a small and diverse representative set of items from a large data set. Diversity here can be modeled in many different ways, and the choice of a diversity measure can significantly affect both practical performance and the algorithmic complexity of finding a diverse set. Both general and application-specific diversity criteria have been proposed in the past~\cite{gong2014diverse,celis2016fair,CevallosEZ17,zhai2015beyond,carbonell1998use}.



In this work, we focus on a popular geometric model of the problem above. While it naturally captures problems in data retrieval and statistics, we show that it also encompasses problems in fair allocation of goods, network design, counting, and optimization. We assume that data are represented as points in the $d$-dimensional Euclidean space, so that choosing a subset of items corresponds to selecting a subset of $d$-dimensional vectors.
A number of natural diversity measures can be formulated in terms of functions of the eigenvalues of the matrix given by the sum of outerproducts of the selected vectors. Some examples are the determinant, the trace, the harmonic mean of the eigenvalues, and the minimum eigenvalue. In this work, we focus on the determinant as the diversity measure. We study the determinant maximization problem with general combinatorial constraints which makes the model rich enough to include many of the problems mentioned above. In particular, we consider matroid constraints, which capture cardinality constraints,  partition constraints, and many more as special cases. This allows modeling constraints imposed by, e.g., budget, feasibility, or fairness considerations.

%

\vspace{2mm}

In an instance of the \DetMax\ problem (under a general matroid constraint), we are given a set of \(n\) vectors
\(\bv_1,\ldots,\bv_n\in\R^d\) and a matroid \(\M=([n],\I)\) with set of
bases $\B$, and our goal is to find a set $S\in \B$ that maximizes $\det\pr{\sum_{i \in S} \bv_i \bv_i^\top}$, i.e.
\begin{equation} \label{eq:D-obj}
\max\left\{\det\pr{\sum_{i\in S} \bv_i \bv_i^\top}: S \in \B\right\}.
\end{equation}
We denote by $k$  the rank of the matroid $\M$, which is the size of all the bases in $\B$. We denote the combinatorial optimization problem \eqref{eq:D-obj} by \DOPT{} and its optimum value by \OPT.


A number of special cases of \DOPT{} have been studied, in which either the choice of vectors or the matroid is restricted~\cite{Welch82,BouhtouGS10,WangYS16,ZeyuanYAY17,SinghX18}. We highlight two illustrative examples. Under cardinality constraints, in which $\B$ consists of all subsets of $[n]$ of size $k$, the problem is hard to approximate to a factor better than $(1+c)^d$ for some $c>0$ when $k=d$~\cite{Koutis06,CivrilM13,SummaEFM15}, and Nikolov~\cite{Nikolov15} gave an $e^d$-approximation for $k\leq d$.\footnote{For $k<d$, the objective is naturally replaced by the product of the  $k$ highest eigenvalues of the matrix, rather than the determinant, which is the product of all $d$ eigenvalues} Interestingly, when $k> d$, improved guarantees are known~\cite{WangYS16,ZeyuanYAY17,SinghX18} with the current best $(1+\epsilon)^d$-approximation when $k\geq d+\frac{d}{\epsilon}$~\cite{MadanSTU19}.

For general matroids, a series of works~\cite{NikolovS16,anari2017generalization,straszak2017real,AnariGV18} have focused on the case when $k\leq d$, and the latest results of Anari, Oveis-Gharan, and Vinzant~\cite{AnariGV18} imply an $e^{2k}$-estimation algorithm. These results were first proved for the special case when the \emph{generating polynomial} for the matroid is a \emph{real stable polynomial}~\cite{anari2017generalization}. Recent and exciting advances on completely log-concave polynomials~\cite{AnariGV18} (and the equivalent notion of Lorentzian polynomials~\cite{branden2019lorentzian}) allow the techniques of~\cite{anari2017generalization} to be generalized to all matroids. While these results are not stated when
$k> d$, the analysis naturally yields an $e^{2k}$-estimation algorithm even in that case. Such a dependence on $k$ is often exorbitant since $k$ can be much larger than $d$ in many applications. Moreover, the hardness result mentioned above only shows that the approximation factor needs to depend exponentially on $d$, but not necessarily on $k$.\footnote{Since the objective is the determinant of $d\times d$ matrices, and the determinant is homogeneous of degree $d$, exponential dependence on $d$ is an appropriate scaling.} A starting point for this work is a result showing that these existing techniques are incapable of removing the dependence on $k$ for general matroid constraints.  More formally, we show in Appendix~\ref{sec:no_pc_dist} that any algorithm which solves a convex relaxation and rounds the fractional solution without using the structure of the vectors yields an approximation factor necessarily dependent on $k$ even when $d=2$.

\subsection{Our Results and Contributions}

Our main result is an algorithm that estimates the objective of the \DetMax{} problem under a general matroid constraint.
\begin{theorem}\label{thm:main}
There is an efficiently computable convex program whose objective value estimates the objective of the \DetMax{} problem under a general matroid constraint within a multiplicative factor of $(O(d))^d$.
\end{theorem}

As outlined earlier, an approximation factor depending only on $d$ cannot be obtained by rounding an arbitrary optimal solution to any of the known convex relaxations of the problem. Our work introduces two key ideas to bypass this bottleneck. First, we show that there always exists an optimal \emph{sparse} fractional solution to a particular convex programming relaxation.
In particular, we show that there always exists an optimal solution with no more than $O(d^2)$ fractional variables out of a total of $n$ variables. The proof of this fact relies crucially on the first order optimality conditions of the convex program. A straightforward presentation of the first order optimality conditions leads to a system of (exponentially many) \emph{non-linear} constraints over an exponential number of variables. We interpret these constraints using matroid theory and reformulate them as a system of (exponentially many) \emph{linear} inequalities. Then, we apply combinatorial optimization techniques such as uncrossing in order to show that any basic feasible solution to the system of inequalities must be sparse, again using the inherent matroid structure of the linear constraints.

Second, we give a new randomized algorithm that rounds such a sparse solution for any matroid, giving  the desired result. Our algorithm crucially uses the near-integral structure of optimal solutions, and thus differs significantly from previous rounding algorithms, which are oblivious to any such structure. The main challenge in the design of the algorithm is that the non-linearity of the objective function implies that even an integral variable cannot be included in the solution with probability 1. Our rounding proceeds in two phases: we first randomly round the fractional variables, and then we randomly choose which of the integral variables to include in a solution, while maintaining feasibility. We again rely on matroid theory to show that the random solution obtained has  large objective value in expectation.

This combination of techniques from convex optimization and matroid theory, which we use in order to find a sparse optimal solution of a convex program with exponentially many constraints, appears to be novel and may be of independent interest.

We also consider the special case of partition matroids due to its significant applications and note that an improved approximation algorithm can be obtained for this case. We observe that the roadblock in achieving an approximation factor independent of $k$ for general matroids does not appear in the case of partition matroids. Thus, the standard randomized rounding algorithm also achieves $e^{O(d)}$-approximation by generalizing the results on Nash Social Welfare in~\cite{anari2018nash}. We include the proof in Theorem~\ref{thm:partition} in the Appendix for completeness.

\paragraph{Deterministic Algorithms.} A challenge for the \DetMax{}  problem under a general matroid constraint has been the lack of \emph{true} approximation algorithms that achieve the same guarantees as the estimation algorithms. Most results~\cite{NikolovS16,anari2017generalization,AnariGV18,AnariLGV19,straszak2017real} give randomized algorithms whose guarantees hold in expectation and are not known to hold with high probability or deterministically. The few existing efficient algorithms with high probability or deterministic guarantees either work only for restricted classes of matroids, such as uniform matroids~\cite{Nikolov15,allen2017near,SinghX18} or partition matroids with a constant number of parts~\cite{CelisDKSV17}, or rely on special structure of the input vectors (or both)~\cite{anari2016nash,cole2016convex,celis2016fair,barman2018finding}. Ebrahimi, Straszak and Vishnoi~\cite{Ebrahimi17} gave the most general algorithmic results that apply to all regular matroids, but the approximation factors they achieved depend on the size of the ground set and not just the dimension of vectors, as aimed in our work.

We utilize the existence of sparse optimal solutions to our convex programming relaxation to give an efficient deterministic algorithm achieving an approximation factor that only depends on the dimension $d$ of the vectors, and not on the size $k$ of the output set or the size $n$ of the input.

\begin{theorem}\label{thm:deterministic}
  There is a polynomial time deterministic algorithm for the \DetMax{} problem that gives an $\left(O(d)\right)^{d^3}$-approximation.
\end{theorem}

The above result is achieved by using the optimal objective value of the convex program as an estimate of the value of an optimal solution, and reducing the search problem of finding an approximately optimal solution to estimation. We have shown that some optimal solution to the convex program has at most $O(d^2)$ fractional variables, and, therefore, has support of size $k + O(d^2)$. Then, producing a feasible solution (which has size $k$) requires finding $O(d^2)$ elements of the support of the optimal solution to exclude from the solution: the remaining $k$ elements form the output. Thus, the sparsity allows us to argue that the estimation problem needs to be recursively solved only $O(d^2)$ times, which is crucial in guaranteeing an approximation factor that depends only on $d$.

We remark the guarantee is worse than is achieved (in expectation) by the randomized algorithm.
Obtaining true approximation algorithms that match the performance of the estimation algorithms remains a challenging open problem for the \DetMax{} problem under a general matroid constraint, even in the case of a partition constraint.


\subsection{Applications}
As mentioned earlier, \DetMax{} models problems in many different areas and our results imply new approximations for many of these problems.  We give details for some of them below.

\paragraph{Experimental Design.} In the optimal experimental design problem for linear models, the goal is to infer an unknown $\btheta^\star \in \R^d$ from a possible set of linear measurements of the form $y_i=\bv_i^\top \btheta^\star +\eta_i$. Here, $\bv_1,\dots,\bv_n \in \R^d$ are known vectors,  and $\eta_1,\dots,\eta_n$ are independent Gaussian noises with mean $0$ and variance $1$. In some settings, performing all of the $n$ measurements might be infeasible, and  combinatorial constraints such as matroid constraints can be used to define the
feasible sets of measurements. Given a set $S\subseteq [n]$ of measurements, an estimator $\widehat{\btheta}$ for $\btheta^*$ is obtained via solving the least squares regression problem $\min_{\btheta \in \R^d} \sum_{i\in S}(y_i-\bv_i^\top\btheta)^2$. The error $\widehat{\btheta}-\btheta^\star$ is distributed
as a $d$-dimensional Gaussian $N\left(0,\left(\sum_{i\in S} \bv_i \bv_i^\top \right)^{-1}\right)$. Minimizing the volume of the confidence ellipsoid, or equivalently the determinant of the covariance matrix of the error, is referred to as $D$-optimal design in statistics~\cite{pukelsheim2006optimal}. Our results directly imply improved approximability for  $D$-optimal design  under a general matroid constraint.

\paragraph{Nash Social Welfare.}
In the indivisible goods allocation problem the goal is to allocate, i.e. partition, $m$ goods among $d$ agents so that some notion of social welfare and/or fairness is achieved. Each agent $i$ has  utility $u_i(j)$ for good $j \in [m]$, and if $S_i$ are the goods assigned to agent $i$, then her utility is $u_i(S_i) = \sum_{j   \in S_i}{u_i(j)}$.
A well studied objective in this context is Nash social welfare (NSW), which asks to maximize $\left(\prod_{i =     1}^du_i(S_i)\right)^{1/d}$. This objective interpolates between maximally efficient and maximally egalitarian allocations  -- see ~\cite{moulin,CKMPSW16} for more extensive background. Maximizing the NSW can be formulated as an instance of \DetMax{} under a partition constraint, as observed in~\cite{anari2016nash}. For each agent $i$ and good $j$, we create a vector $\bv_{(i,j)} = \sqrt{u_i(j)} \be_i$, where $\be_i$ is the $i$-th standard basis vector of $\R^d$, and form a partition matroid $\M$ whose bases $\B$ consist of all sets $S \subseteq [d] \times [m]$ such that $|\{i: (i,j) \in S\}| = 1$ for all $j \in [m]$. Then, a feasible solution $S \in \B$ corresponds to an allocation of the goods, and the determinant $\det\left(\sum_{(i,j)  \in S} \bv_{(i,j)} \bv_{(i,j)}^\top\right)$ is equal to the NSW objective. Our results recover those in~\cite{anari2016nash} and further allow us to give an $O(d)$-estimation algorithm when the allocation $(S_1, \ldots, S_d)$ is required to satisfy additional matroid constraints.  For example, the works~\cite{GourvesMT13,GourvesMT14,GourvesM19} considered allocations such that $\bigcup_{i = 1}^d S_i$ is a basis of a matroid $\M'$. We can model this setting by defining our constraint matroid $\M$ so that $S \subseteq [d] \times [m]$ is a basis of $\M$ if and only if $|\{i: (i,j) \in S\}| = 1$ for all $j \in [m]$ and $\{j: \exists i\text{ s.t. } (i,j) \in S\}$ is a basis of $\M'$. Our results then imply an $O(d)$-estimation algorithm and an $O(d)^{d^2}$-approximation algorithm for maximizing NSW subject to these general matroid constraints.


\paragraph{Network Design Problems.}
In general, the goal in network design problems is to pick a subset $F$ of the edges of an undirected graph $G=(V,E)$ with non-negative edge weights $w$ such that the subgraph $H=(V,F)$ is \emph{well-connected}. One measure of  connectivity is to maximize the total weight of spanning trees in $H=(V,F)$, where the weight of a tree is defined as the product of the weights of its edges (see~\cite{LiPYZ19} and references therein for other applications). This natural network design problem is a special case of the \DetMax{} problem. For each $(i,j) \in E$, we introduce a vector $\bv_{(i,j)}\in \{0,1,-1\}^V$  with $(\bv_{(i,j)})_i=\sqrt{w_{(i,j)}}$, $(\bv_{(i,j)})_j=-\sqrt{w_{(i,j)}}$, and the rest of the coordinates set to zero. Observe that $\sum_{e\in F} \bv_e \bv_e^\top$ is exactly the Laplacian of $H=(V,F)$, and the determinant of the Laplacian\footnote{We remark that the Laplacian is always singular, but we can first project the vectors $\bv_e$ orthogonal to the all-ones vector and take the determinant in $d-1$ dimensions.}
~gives the number of spanning trees in $H$. Our results imply an $O(|V|)^{|V|}$-estimation algorithm, and $O(|V|)^{|V|^3}$-approximation algorithm for this problem under a general matroid constraint.

\subsection{Technical Overview}

Our starting point is a variant of the convex relaxation introduced in~\cite{NikolovS16} for the partition matroid. Let the set of input vectors be $V=\{\bv_1,\dots,\bv_n\}\subset \R^d$.  For a matroid \(\M=([n],\I)\), we denote by
\(\I_s(\M):=\set{S\in\I:|S|=s}\) the set of all independent sets of
size \(s\). We denote by \(\P(\M)\) the matroid base polytope of
\(\M\), which is the convex hull of the indicator vectors of the bases. For any vector \(\bz\in\R^{n}\)
and a subset \(S\subseteq [n]\), we let \(z(S):=\sum_{i\in S}
z_i\). We let $\Z := \{ \bz \in \R^{n} : \forall S \in \I_d(\M), z(S) \geq
0\}$. Our convex relaxation is
\begin{equation}
\sup_{\bx\in\P(\M)}  \inf_{\bz\in\Z} g(\bx,\bz) := \log\det\pr{\sum_{i \in [n]} x_ie^{z_i} \bv_i \bv_i^\top}. \label{eq:obj-CP-intro}
\end{equation}
For ease of notation, we define $f(\bx):=\inf_{\bz\in\Z} g(\bx,\bz)$, the inner infimum of \eqref{eq:obj-CP-intro}.

 Similar but somewhat different convex programs have been studied by ~\cite{anari2016nash,anari2017generalization,straszak2017real,straszak2017belief}. (The relationship of our convex program to these also plays a crucial role in our analysis: see below.) The estimation algorithms in these works rely on a simple randomized algorithm to round a fractional optimal solution $\bx^\star$. The analysis of the algorithm relies on a positive correlation property: the algorithm outputs a random solution such that all elements of an independent set $S$ of size $d$ are included with probability at least $\frac{1}{\alpha} \cdot \prod_{i \in S} x_i^*$, where $\alpha$ is some function of $k$. This property, combined with inequalities for real stable and completely log-concave polynomials, leads to an $\alpha \cdot e^{O(d)}$-estimation algorithm. We show that there exist fractional optimal solutions $\bx^\star$ such that no rounding scheme has this positive correlation property for any $\alpha$ which is a function of $d$ and independent of $k$. So, the dependence on $k$ is inherent to all the previous algorithms which round an arbitrary optimal solution $\bx^\star$ and do not consider the structure of the vectors to obtain some structure on the optimal $\bx^\star$.


Our first technical result is to show that there always exists an optimal solution that has at most $O(d^2)$ fractional variables. We briefly describe how to obtain such a sparse optimal solution. Let $\bx^\star$ denote an optimal solution to the convex program (similar reasoning works for near optimal solutions as well). We first show that, using a series of careful preprocessing steps, we can assume that there exists a $\bz^\star$ attaining the infimum in $f(\bx^\star) = \inf_{\bz\in \Z} g(\bx^\star,\bz)$. We then use first order optimality conditions that give a sufficient condition for another solution $\bx$ to be optimal (i.e., to have $f(\bx) = f(\bx^\star)$). These conditions, however, present two significant  obstacles: first, the conditions are not linear in $\bx$, and, second, they ask for the existence of an exponentially sized dual solution as a certificate of optimality. We address the first problem by noticing that insisting that the entire matrix $\pr{\sum_{i \in [n]} x_i^\star e^{z_i} \bv_i \bv_i^\top}$ does not change when $\bx^\star$ changes to $\bx$ leads to the optimality conditions becoming a system of linear equations in exponentially many variables. We then use the simple, yet elegant fact from matroid theory that minimum weight bases of a matroid under a linear weight function form the base set of another matroid. We use this combinatorial fact to observe that the existence of the exponentially sized dual solution is equivalent to insisting that a vector, whose coordinates are linear functions of $\bx$, is in the base polytope of a new matroid. Putting all of this together reduces the search for the new optimal solution $\bx$ to solving a system of exponentially many linear inequalities. Now, in the familiar territory of matroid polytopes, we apply standard uncrossing methods and show that every extreme point solution of the system of these linear inequalities has only $O(d^2)$ fractional variables.

Finally, we give a new randomized algorithm that gives an $O(d)^d$-estimation algorithm in the presence of $O(d^2)$ fractional variables. Since the objective is non-linear, we cannot just pick all variables set to 1 and apply a randomized algorithm to fractional elements. Indeed, the variables set to 1 must also be dropped from the final solution with certain probability. We show that given a solution $\bx$ with at most $O(d^2)$ fractional values, our rounding scheme outputs a random solution such that for any independent set $S$ of size $d$, all elements of $S$ are picked with probability at least $\pr{O(d)}^{-d}\prod_{i \in S} x_i$. To show that this property implies the random solution output by the algorithm achieves an $O(d)^d$ approximation in expectation, we utilize recent and exciting work on strongly log-concave polynomials~\cite{AnariGV18,AnariLGV19} and the equivalent notion of Lorentzian polynomials~\cite{branden2019lorentzian}. While the analysis using strongly log-concave distributions naturally utilizes a different convex programming relaxation introduced in \cite{anari2017generalization}, the aforementioned sparsity result is not applicable to these convex programs. To this end, we show that the convex programming relaxation considered in our work is stronger than the convex programming relaxation from \cite{anari2017generalization}. The relationship between the various convex programs for this problem and their respective strengths and weaknesses outlined by our results may be of independent interest.

\subsection{Related Work}

Below we given an overview of prior work on the \DetMax{} problem, which has been studied in many special cases.

\paragraph{Uniform Matroid:} \DetMax{} is NP-hard even for a uniform matroid ~\cite{Welch82}. Koutis~\cite{Koutis06} showed that there exists a constant $c>0$ such that it is NP-hard to achieve approximation better than a factor of $(1+c)^d$ when $k \le \beta d$ for some constant $\beta < 1$, and Di Summa et al.~\cite{SummaEFM15} extended this hardness result to $k=d$.
Bouhtou \etal ~\cite{BouhtouGS10} gave an $\pr{\frac{n}{k}}^d$-approximation algorithm based on rounding the solution of a natural convex relaxation. Nikolov~\cite{Nikolov15} improved the result to an $e^{k}$-approximation when $k \leq d$. Wang \etal~\cite{WangYS16} improved the approximation ratio to $(1+\epsilon)^d$ when $k \geq \frac{d^2}{\epsilon}$. Allen-Zhu \etal~\cite{ZeyuanYAY17} improved the bound on $k$ to give $(1+\epsilon)^d$-approximation when $k = \Omega\pr{\frac{d}{\epsilon^2}}$ and showed the existence
of a sparse optimal solution for the standard convex relaxation. This was improved by Singh and Xie~\cite{SinghX18} who gave a $(1+\epsilon)^d$-approximation when $k = \Omega\pr{\frac{d}{\epsilon} + \frac{1}{\epsilon^2} \log \frac{1}{\epsilon}}$. Recently, this was improved by Madan \etal~\cite{MadanSTU19} who gave a $(1+\epsilon)^d$-approximation when $k \geq d + \frac{d}{\epsilon}$.

\paragraph{General Matroid:} Nikolov and Singh~\cite{NikolovS16} gave an $e^d$-estimation algorithm for \DetMax{} under a partition matroid of rank $d$.  
Straszak and Vishnoi~\cite{straszak2017real} gave an $O(e^n)$-estimation (where $n$ is the size of the ground set), and Anari and Gharan~\cite{anari2017generalization} gave an $e^{2k}$-estimation when the generating polynomial for the matroid is real-stable. This corresponds to  Strongly Rayleigh matroids which include  uniform and partition matroids. These results were generalized by Anari, Gharan, and Vinzant~\cite{AnariGV18} who gave an $e^{2k}$-estimation for a general matroid.\footnote{While the result in~\cite{AnariGV18} is not stated for $k>d$, it can be easily deduced from the analysis.}
Algorithms in~\cite{NikolovS16,anari2017generalization,AnariGV18} estimate the optimum value within a certain approximation factor, but they do not yield an approximate solution with high probability in polynomial time. For partition and regular matroids of rank $k \leq d$, Ebrahimi, Straszak, and Vishnoi~\cite{Ebrahimi17}, using anti-concentration inequalities, gave efficient approximation algorithms with high probability guarantees. These are the most general algorithmic approximation results known for the \DetMax{} problem. Their guarantees on the approximation factor, however, are worse than the estimation algorithms and depend on the size of the ground set.

\paragraph{Experimental Design:} In the experimental design literature, several different objective functions are studied, which lead to different optimization problems. Apart from $D$-optimal design, two of the most notable problems are $A$-optimal design and $E$-optimal design. In $A$-optimal design, the objective is to minimize the trace of the covariance matrix: $\min_{S \in \B} \tr\pr{\pr{\sum_{i \in S} v_i v_i^\top}^{-1}}$. In $E$-optimal design, the objective is to minimize the maximum eigenvalue of the covariance matrix: $\min_{S \in \B}\lambda_{\max}\pr{\pr{\sum_{i \in S} v_i v_i^\top}^{-1}}$. There have been a series of works on both of these problems in the uniform matroid setting~\cite{avron2013faster,WangYS16,nikolov2018proportional,ZeyuanYAY17,MadanSTU19}. The current best results are $(1+\epsilon)$-approximation for $A$-design  when 
$k\geq \Omega\pr{\frac{d}{\epsilon} + \frac{1}{\epsilon^2} \log \frac{1}{\epsilon}}$~\cite{nikolov2018proportional} and $(1+\epsilon)$-approximation for $E$-design when 
$k\geq \Omega\pr{\frac{d}{\epsilon^2}}$~\cite{ZeyuanYAY17}.

\paragraph{Nash Social Welfare:} Cole and Gkatzelis~\cite{CG15} gave the first constant factor approximation algorithm for the Nash Social Welfare problem, achieving an approximation factor of $(2e^{1/e})$. This result was subsequently improved in a series of papers~\cite{anari2016nash,cole2016convex,barman2018finding} with the current best approximation ratio being $1.45$.

\paragraph{Completely Log Concave Polynomials:} The theory of completely log concave polynomials introduced in ~\cite{AnariGV18,AnariLGV19} (see also ~\cite{Gurvits09-Newton,branden2019lorentzian}) plays an important role in the analysis of algorithms for the \DetMax{} problem. These results build on the use of stable polynomials in the analysis of algorithms in~\cite{NikolovS16,anari2016nash,anari2017generalization,straszak2017real}, themselves building on the results by Gurvits~\cite{Gur06}.

\paragraph{Sparsity and Fractionality in Convex Programs.}
Bounding the number of fractional variables,  the sparsity of optimal solutions of convex programs, and, in particular, of convex relaxations of discrete problems is a powerful technique which appears in many different contexts. In combinatorics and geometry, early examples can be found in the proof of the Beck-Fiala theorem in discrepancy theory~\cite{beck1981integer} and in work of Barany, Grinberg, and Sevastyanov~\cite{Sevast78,GrinbergS80,BaranyG81}. A survey of these results is given by Barany~\cite{Barany08}. In approximation algorithms, an early example is the Karmakar-Karp approximation algorithm for the bin packing problem~\cite{KarmarkarK82}. Bounding the sparsity and fractionality of optimal basic feasible solutions to linear programs is the basis of the iterative rounding method in approximation algorithms, introduced by Jain~\cite{jain2001factor}. The book~\cite{lau2011iterative} gives many results derived from this method. Bounding the sparsity of basic feasible solutions is also key to the linear programming approach in compressed sensing~\cite{CandesTao05}. Related results are known for the matrix completion problem, where sparsity is defined in terms of matrix rank and the corresponding optimization problem is non-linear~\cite{CandesTao10}. Sparsity of optimal solutions of non-linear convex programs appears to be, however, underexplored in general.

\subsection{Organization}
In Section~\ref{sec:solvability}, we discuss our convex relaxation, some technical issues in solving the relaxation, our main technical result, and the first order optimality conditions for the relaxation. In Section~\ref{sec:sparse}, we show the existence of an optimal solution with at most $O(d^2)$ fractional values. In Section~\ref{sec:rounding}, we give the randomized algorithm to round a solution of the relaxation with few fractional values. In Section~\ref{sec:deterministic}, we give our deterministic approximation algorithm that gives a guarantee that only depends on $d$. In Appendix~\ref{sec:prelim}, we discussed some of the definitions and preliminaries related to matroids, log-concavity, and real stability. In Appendix~\ref{sec:preprocessing}, we discuss the preprocessing of a given instance so that the convex relaxation is solvable and the inner infimum is achieved. In Appendix~\ref{sec:KKT}, we derive optimality conditions for our convex relaxation. In Appendices~\ref{sec:proof-sparse}~and~\ref{sec:app-rounding}, we give missing proofs from Sections~\ref{sec:sparse}~and~\ref{sec:rounding}, respectively.
In Appendix~\ref{sec:no_pc_dist}, we show an example proving that none of the previous approaches can achieve an approximation factor independent of $k$. In Appendix~\ref{sec:partition_approx}, we give an improved approximation algorithm for \DetMax\ under a partition matroid.

%% file: solvability.tex
\section{Convex Program and Optimality Conditions}\label{sec:solvability}

Our algorithm for \DetMax{} under a general matroid constraint   is
based on solving a convex relaxation and rounding
an optimal solution of the convex relaxation to an integral solution. In this section,
we formulate this convex relaxation, show that it is efficiently
solvable, and prove some of its properties which are crucial for the
rounding algorithm.

\subsection{Formulation of the Convex Program}

Let $V=\{\bv_1,\dots,\bv_n\}$ be input vectors.  For a matroid
\(\M=([n],\I)\), we denote by
\(\I_s(\M):=\set{S\in\I:|S|=s}\) the set of all independent sets of
size \(s\). We denote by \(\P(\M)\) the matroid base polytope of
\(\M\), which is the convex hull of all of the bases. We include some basic preliminaries on matroids in Appendix~\ref{sec:matroidprelim}. For any vector \(\bz\in\R^{n}\) of real numbers
and a subset \(S\subseteq [n]\), we let \(z(S):=\sum_{i\in S}
z_i\). We let $\Z := \{ \bz \in \R^{n} : \forall S \in \I_d(\M), z(S) \geq
0\}$. We  introduce the optimization problem
\begin{equation} \label{eq:D-relax}
\sup_{\bx\in\P(\M)}  \inf_{\bz\in\Z} g(\bx,\bz) := \log\det\pr{\sum_{i \in [n]} x_ie^{z_i} \bv_i \bv_i^\top}.
\end{equation}
For ease of notation, we also let $f(\bx):=\inf_{\bz\in\Z} g(\bx,\bz)$, the inner infimum of \eqref{eq:D-relax}.
The above program is a convex relaxation, as shown in Nikolov and Singh~\cite{NikolovS16}. We include a proof for completeness in Lemma~\ref{lem:relaxation} in the Appendix. Unfortunately, it is not clear whether the outer supremum and  inner
infimum are attained at some \(\bx^\star\) and finite \(\bz^\star\). While the supremum over $\bx$ can be
approximated, our approach relies crucially on the inner infimum being
achieved exactly at some finite \(\bz^\star\). We first show the following technical lemma that gives a
sufficient condition for the infimum to be achieved based on KKT
conditions and Slater's qualification of constraints. We
say that the vectors $\{\bv_i: i\in [n]\}\subseteq \rR^d$ are in \emph{general
  position} if any subset of size $d$ is linearly independent.

\begin{lemma} \label{lem:finite}
   Let \(\bx\in\P(\M)\) be such that $\max_{i \in [n]} x_i < 1$ and 
  \(f( \bx)=\inf_{\bz\in\Z} g( \bx,\bz)\) is finite, and suppose that
  the vectors $\{\bv_i: i\in [n]\}$ are in general position. Then,
  \(g( \bx,\bz)\) attains its infimum over \(\bz\in\Z\) at some
  \(\bz^*\in\Z\).
\end{lemma}

In general, an instance of our problem may not satisfy the conditions
of the lemma: the given vectors need not be in  general position,
and every optimal $\bx$ may have value $1$ on some coordinates.  We
outline a preprocessing step in Appendix~\ref{sec:preprocessing} to
show that both of these assumptions can be made with a slight loss in
optimality by modifying the input instance. This is achieved by
modifying the matroid by introducing two parallel copies of each
element as well as perturbing the vectors slightly to put them in
general position. From here on, we assume that these modifications
have been carried out, and we use $\M$ and $V$ to denote the resulting
matroid and vectors, respectively.


These reductions allow us to formulate the following stronger convex
program where we place an additional upper bound on the coordinates of
$\bx$: 

\begin{equation} \label{eq:D-relax2}
\sup_{\bx\in\P(\M)\cap \left[0,\frac12\right]^n}  \inf_{\bz\in\Z} g(\bx,\bz):= \log\det\pr{\sum_{i \in [n]} x_ie^{z_i} \bv_i \bv_i^\top}
\end{equation}
We denote the convex program \eqref{eq:D-relax2} by \CPD{}, its optimum value by \CP, and an optimal solution by
\((\bx^\star,\bz^\star)\). We denote by $\OPT{}$, the optimal value of the
\DetMax{} problem.
Based on the discussion above, we show the
following lemma where we also outline the polynomial time solvability
of the convex program. The proof of the lemma appears in
Appendix~\ref{sec:preprocessing}.

\begin{lemma}\label{lem:alg}
  For any $\epsilon >0$, there is a polynomial time algorithm that
  returns  ${\bx^\star}\in \P(\M)\cap \left[0,\frac12\right]^n$ such
  that $\inf_{\bz\in\Z} g( \bx^\star,\bz) \geq \log
  \left(\OPT{}\right) -\epsilon$. Moreover, there exists 
  $\bz^\star$ attaining the infimum in $\inf_{\bz\in\Z} g(
  \bx^\star,\bz)$.
\end{lemma}

Our main algorithmic result is to show that the value of the convex program \CP{}
gives a good approximation of the optimal value $\OPT{}$ of the
\DetMax{} problem. The theorem below immediately implies Theorem~\ref{thm:main}.

\begin{theorem}\label{thm:main2}
The optimum value \CP{} of the convex program gives a $(2e^5d)^{d}$-approximation to the value of the optimum, i.e.,
\begin{align}
\log \left(\OPT{}\right) -\epsilon  \leq \CP{}\leq  \log \left(\OPT{}\right) + O(d\log  d). \label{eq:thm-OPT-CP}
\end{align}
Moreover, there is a polynomial time algorithm that, given $\bx^\star$
attaining $\CP{}$ and $\bz^\star$ attaining the infimum in
$\inf_{\bz\in\Z} g( \bx^\star,\bz)$, returns a random set $S\in \I$
such that
$$\E\left[\det\left(\sum_{i\in S} \bv_i \bv_i^T \right)\right]\geq (2e^5d)^{-d} \left(\OPT{}\right).$$
\end{theorem}

We now outline the ideas behind proving  Theorem~\ref{thm:main2}.
First, we obtain the KKT optimality conditions of $\inf_{\bz \in \Z}
g(\bx,\bz)$ in Section~\ref{sec:first-order}. In
Section~\ref{sec:sparse}, we show that the KKT conditions can be
related to a new matroid defined by minimum weight bases of the
original matroid under the weight function $\bz^\star$. We then apply
uncrossing methods on matroids to show that there is always an optimal sparse
solution -- in particular, one with at most $O(d^2)$ fractional
variables. In Section~\ref{sec:rounding}, we give a rounding algorithm
that uses the fact that number of fractional variables is bounded, and we
prove Theorem~\ref{thm:main2} building on inequalities proved
in~\cite{anari2017generalization}~and~\cite{AnariGV18} for
stable and completely log concave polynomials, respectively.

\subsection{Optimality Conditions}\label{sec:first-order}

Recall the notation \(f( \bx)=\inf_{\bz\in\Z} g( \bx,\bz)\). 
In the following result, we state a sufficient condition that some
feasible solution $\hat{\bx} \in \P(\M)$ satisfies $f(\hat{\bx})
=f(\bx^\star)$, where $\bx^\star$ is an (approximately) optimal
solution to \CPD{} as returned by the algorithm in
Lemma~\ref{lem:alg}. The result is obtained by applying the general
KKT conditions to the optimization problem $\inf_{\bz \in \Z}g(\bx,\bz)$. For completeness, we give a
detailed description of the general KKT conditions and Slater's constraint qualification in
Appendix~\ref{sec:gen-KKT}.  

\begin{lemma} \label{thm:KKT} Suppose \(\bx^\star\in\P(\M)\cap \left[0,\frac12\right]^n\) is a feasible
  solution for \CPD{} such that the infimum over \(\Z\) in \CPD{} is achieved, and
  let \(\bz^\star\in\argmin_{\bz\in\Z} g(\bx^\star,\bz)\). For any \(\hat \bx\in \P(\M)\),
  suppose that there exists
  \(\blambda\in\R^{\I_d(\M)}_{\geq0}\) such that
\begin{enumerate}[label={\arabic*.}]
\item
for all \(S\in \I_d(\M)\) with \(\bz^\star(S) \neq 0\), we have \(\lambda_S = 0\),
\label{cond:CS}
\item
for all \(i\in[n]\), we have \(\hat x_i e^{z_i^\star}\bv_i^\top \bX^{-1} \bv_i = \sum_{S\in \I_d(\M): i\in S}\lambda_S\) where \(\bX=\sum_{i=1}^n x_i^\star e^{z_i^\star}\bv_i \bv_i^\top\), and
\label{cond:first-order}
\item
\(\sum_{i=1}^n x_i^\star e^{z_i^\star}\bv_i \bv_i^\top = \sum_{i=1}^n \hat x_i e^{z_i^\star}\bv_i \bv_i^\top\).
\label{cond:matrix}
\end{enumerate}
Then, \(f(\hat \bx)= f(\bx^\star)\). Moreover,
there exists
\(\blambda\in\R^{\I_d(\M)}_{\geq 0}\) such that
the above three
conditions hold with \(\hat{\bx} = \bx^\star\).
\end{lemma}

We remark that the above criteria ask for the existence of
exponentially sized vector $\blambda$ in order to certify that $\hat
\bx$ is optimal. In the next section, we show that the above condition
is equivalent to showing a certain vector is in the base polytope of
another matroid derived from $\M$.



%% file: sparsity.tex
\section{Small Support Solutions to \CPD}\label{sec:sparse}
\subsection{Preserving the Value of a Solution}
In this section, we show that there is always an optimal solution to \CPD{} that has small number of fractional components. Indeed, given any solution \(\bx\) such that the inner infimum of \CPD{} is attained, we show how to obtain a sparse solution whose objective is no worse.

\begin{theorem}[Sparsity of an optimal solution]\label{thm:sparsity}
Let $\bx^\star$ be a solution to \CPD{} such that the inner infimum of \CPD{} is attained. Then there exists a solution \(\hat {\bx}\in \P(\M)\) such that
\begin{enumerate}
\item $f(\hat {\bx})= f(\bx^\star)$, and
\item $\abs{\set{i\in[n]: 0 < \hat x_i  < 1}} \leq 2\left(\binom{d+1}{2}+d\right).$
\end{enumerate}
Moreover, such a solution $\hat{\bx}$ can be found in polynomial time.
\end{theorem}

\begin{proof}
Given $\bx^\star$, a solution to \CPD{}, we let $\bz^\star$ be an optimal solution to $\inf_{\bz\in \Z} g(\bx^\star, \bz)$.
Also, let $\bX=\sum_{i\in[n]} x_i^\star e^{z_i^\star}\bv_i \bv_i^\top$. We  assume that $\supp{\bx^{\star}}=\set{1,\ldots,n}$ since for any $i$ with $x^\star_i=0$, we can update the instance by deleting these elements. Observe that this does not effect the optimality (restricted to $\supp{\bx^\star}$)
of $\bz^\star$  (see Lemma \ref{lem:restrict-U} in the Appendix for details).

We first give a simpler description than Lemma~\ref{thm:KKT} for a solution $\hat {\bx}$  to have an objective better than $f(\bx^\star)$. This relies on the following basic lemma.
\begin{lemma} \label{lem:sparse-another-matroid}
Let
$\B^\star = \set{S\in \I_d: z^\star(S)=0}$.
Then, $\B^\star$ is a basis of another matroid $\M^\star=([n], \I^\star).$ Additionally, if $\M$ admits an independent oracle, then $\M^\star$ also admits an independent oracle.
\end{lemma}
\begin{proof}
Since $z^\star(S)\geq 0$ for all $S\in \I_d$, the basis of $\I_d$ included in $\I^\star$ are the minimum weight bases under the weight function $\bz$. Minimum weight bases of a matroid form the bases of another matroid, and the independence oracle can be implemented in polynomial time (see Lemma~\ref{lem:bases} in the Appendix for details).
\end{proof}

We now have the following simpler description for $\hat \bx$ to be optimal building on Lemma~\ref{thm:KKT}. Let \(\M^\star\) be the matroid in Lemma \ref{lem:sparse-another-matroid} and let $r^\star: 2^{[n]}\rightarrow \zZ_+$ denote the rank function of $\M^\star$.
\begin{lemma}
Let $\bx^\star$ be a solution of  \CPD{} and $\bz^\star \in \argmin_{\bz\in \Z} g(\bx^\star, \bz)$. Let $\hat \bx\in \R^{[n]}$ be such that
\begin{enumerate}
\item $\hat \bx \in \P(\M)$,
\item the vector $\bw\in \R^{[n]}$ defined as $w_i= \hat{x}_i e^{z_i^\star}\bv_i^\top \bX^{-1} \bv_i$ for each $i\in [n]$ satisfies $\bw \in \P(\M^\star)$, where $\bX=\sum_{i\in[n]} x_i^\star e^{z_i^\star}\bv_i \bv_i^\top$,
\item \(\sum_{i\in[n]} \hat x_i  e^{z_i^\star}\bv_i \bv_i^\top = \sum_{i\in[n]} x_i^\star e^{z_i^\star}\bv_i \bv_i^\top\), and
\item $\supp{\hat \bx}\subseteq \supp{\bx^\star}$.
\end{enumerate}
Then $f(\hat{\bx})= f(\bx^\star)$.
\end{lemma}
\begin{proof}
We show that the above conditions imply that the conditions of Lemma~\ref{thm:KKT} are satisfied. Indeed, we only need to show the existence of $\blambda \in \R^{\I_d(\M)}$ as claimed. Since $\bw\in \R^{[n]}$ is in  $\P(\M^\star)$, we have $\bw=\sum_{S\in \B(\M^\star)} \mu_S \chi_S$ where $\chi_S\in \R^{[n]}$ is the indicator vector of set $S$ and $\sum_{S\in \B(\M^\star)} \mu_S=1$. Observe that for each $S\in \B(\M^\star)$, we have $z^\star(S)=0$. Thus, setting $\lambda_S=\mu_S$ for $S\in \B(\M^\star)$ and $\lambda_S=0$ for all other sets in $\I_d(\M)$ satisfies the conditions of Lemma~\ref{thm:KKT}.
\end{proof}

\def \lp {LP}
\begin{figure}
\centering
\begin{boxedminipage}{0.6\textwidth}
\begin{align}
\label{LP:start} &\min 0 \qquad   \\
\label{mat11}s.t. \qquad \qquad \sum_{i\in S} x_i &\leq r(S) \qquad \ \forall\ \emptyset \subsetneq S\subsetneq [n] \\
\label{mat12}x([n]) &=r([n]) =k\\
\label{mat21}\sum_{i\in S} x_i  e^{z_i^\star} \bv_i^\top \bX^{-1} \bv_i &\leq r^\star(S) \qquad \forall\ \emptyset \subsetneq S\subsetneq [n]\\
\label{mat22}\sum_{i\in [n]} x_i  e^{z_i^\star} \bv_i^\top \bX^{-1} \bv_i &= r^\star([n])=d \\
\label{vec1}\sum_{i=1}^n x_i e^{z_i^\star} \bv_i \bv_i^\top &= \sum_{i=1}^n  x_i^\star  e^{z_i^\star} \bv_i \bv_i^\top \\
x_i &\geq 0 \qquad \qquad \forall i\in [n] \label{LP:x-geq-0}
\end{align}
\end{boxedminipage}\caption{Linear program to obtain a sparse solution.}\label{fig:sparse}
\end{figure}

Now the above conditions can be formulated as a feasibility system over the following linear constraints as given in Figure~\ref{fig:sparse},  and we call the formulated linear program \LP. Here, constraints \eqref{mat11}-\eqref{mat12} insist that $\bx\in \P(\M)$ and \eqref{mat21}-\eqref{mat22} insist that the vector $(x_i  e^{z_i^\star} \bv_i \bX^{-1} \bv_i )_{i\in [n]} \in \P(\M^\star)$. Constraints \eqref{vec1} insist that the matrix $\bX$ does not change when the solution changes to $\bx$ from $\bx^\star$. For ease of notation, we let $\bw_x$  be the vector $(x_i  e^{z_i^\star} \bv_i \bX^{-1} \bv_i )_{i\in [n]}$.

From basic uncrossing methods we obtain the following lemma characterizing any extreme point of the above linear program. Recall that a collection \(\C\) of sets is a \textit{chain} if for all \(A,B\in \C\), we have \(A\subseteq B\) or \(B\subseteq A\). The proof of the lemma appears in Appendix \ref{sec:proof-sparse}. Again, we focus on $\supp{\bx}$ since \(\bx\) remains extreme after removing coordinates with $x_i=0$. Thus, we assume that $[n]=\supp{\bx}$.
\begin{lemma} \label{lem:extreme-LP-uncross}
If $\bx$ is an extreme point of the linear program \LP, then there
exist chains $\C_1,\C_2 \subseteq 2^{[n]}$ and $P\subseteq [d]\times[d]$ such that
\begin{enumerate}
\item $x(S)=r(S)$ for each $S\in \C_1$,  $w_x(S)=r^\star(S)$ for each $S\in \C_2$, and  $(\sum_{i=1}^n x_i e^{z_i^\star}\bv_i \bv_i^\top)_{jk}=(\sum_{i=1}^n  x_i^\star e^{z_i^\star} \bv_i \bv_i^\top)_{jk}$ for each $(j,k)\in P$,
\item the linear constraints corresponding to sets in $\C_1, \C_2$ and pairs in $P$ are linearly independent, and
\item $|\supp{\bx}|=|\C_1|+|\C_2|+|P|$.
\end{enumerate}
\end{lemma}
Let $\bx$ be an extreme point of the linear program \LP. Such an $\bx$ can be found in polynomial time.  Let $\C_1=\{S_1,\ldots, S_l\}$ where $S_1\subset S_2 \ldots \subset S_l$. Then, we have $x(S_i)=r(S_i)$. Since $x_i>0$ for all $i\in [n]$, we have $1\leq r(S_1)<r(S_2)\ldots <r(S_l)\leq k$ and from the integrality of the rank function, we obtain that  $|\C_1|=l\leq k$. Similarly, $|\C_2|\leq r^\star([n])\leq d$, and clearly $|P|\leq \binom{d+1}{2}$ since $\sum_{i=1}^n x_i e^{z_i^\star}\bv_i \bv_i^\top$ and \(\sum_{i=1}^n  x_i^\star e^{z_i^\star} \bv_i \bv_i^\top\) are $d\times d$ symmetric matrices. Therefore, $\supp{\bx}\leq k+d+\binom{d+1}{2}$. In what follows we argue all but $2\left(d+\binom{d+1}{2}\right)$ coordinates are set to 1.

For ease of notation, we let $S_0=\emptyset$. Observe that if $|S_j\setminus S_{j-1}|= 1$ for any $1\leq j\leq l$, say $\{i\}=S_j\setminus S_{j-1}$, then $x_i=x(S_j)-x(S_{j-1})= r(S_j)-r(S_{j-1})$ which is an non-negative integer. Since $x_i>0$, we obtain that $x_i=1$. Let $I=\{1\leq j\leq k: |S_j\setminus S_{j-1}|=1\}$. Observe that there are at least $|I|$ variables set to $1$. But since every set $S_j$ with $j\notin I$ contains at least two elements in $S_j\setminus S_{j-1}$, we have
$$|\supp{\bx}|\geq |I|+2(l-|I|).$$
But from Lemma~\ref{lem:extreme-LP-uncross}, we have$$|\supp{\bx}|\leq l+d+\binom{d+1}{2}.$$

Combining the two inequalities, we get $|I|\geq l-d-\binom{d+1}{2} \geq |\supp{\bx}|-2(d+\binom{d+1}{2})$. Hence, the number of fractional variables is at most $\supp{x}-|I|\leq 2(d+\binom{d+1}{2})$.
\end{proof}

%% file: rounding.tex
\section{Randomized Rounding Algorithm}\label{sec:rounding}
In this section, we give our randomized rounding algorithm and prove the
guarantee on its performance claimed in Theorem~\ref{thm:main2}.


Throughout this section, we assume that the algorithm receives an input  \(\bx\in \P(\M)\) such that
 \begin{displaymath}
\abs{\set{i:0 < x_i < 1}} \leq 2\pr{{d+1\choose 2} +d}.
\end{displaymath}

We first describe the rounding algorithm, presented in Algorithm~\ref{alg:rounding}. It is obvious that Algorithm~\ref{alg:rounding} runs in polynomial time.
\begin{algorithm}
\caption{Rounding Algorithm}
\begin{algorithmic}[1]
\State \textbf{Input:} a matroid \(\M=([n],\I)\), \(\bx\in\P(\M)\).
\State \textbf{Output:} a set \(S\in\I\).
\Procedure{Rounding}{$x,\I$}
        \State \(R_1 \leftarrow \set{i:0<x_i<1}, R_2\leftarrow\set{i:x_i=1}\)
        \State \(T\leftarrow \emptyset\)
        \For{\(i\) in \(R_1\)}
               \If{$T \cup \set{i} \in \I$}
                     \State \(T\leftarrow T\cup\set{i}\) with probability \(\frac 1d\)
               \EndIf
        \EndFor
        \For{\(i\) in \(R_2\)}
               \If{$T \cup \set{i} \in \I$}
                     \State \(T\leftarrow T\cup\set{i}\) with probability \(\frac 12\)
               \EndIf
        \EndFor
        \If{\(T\) is not a basis}
                \State Extend \(T\) to a basis (e.g. by going through each element in \([n]\setminus T\) and add it to \(T\) if \(T\) remains independent until \(T\) is a basis)
                \EndIf
        \Return $T$
        \EndProcedure
\end{algorithmic}
\label{alg:rounding}
\end{algorithm}

For ease of notation we denote $\gamma= (2e^3d)^{-d}$ and \(\I_d=\I_d(\M)\). We first claim that every independent subset $S$ of $R_1 \cup R_2$ of size $d$ is contained in the output set with probability at least $\gamma$. The claim can only be true if the ground set $R_1 \cup R_2$, which has been restricted to the support of $\bx$,  is small. 

\begin{lemma}\label{lem:rounding_prob}
Let $T$ denote the random set returned by Algorithm~\ref{alg:rounding}. Then, for any set $S \subseteq R_1 \cup R_2$ such that $S\in \I_d$, we have
$$\Pr[S\subseteq T]\geq \gamma.$$
\end{lemma}

Lemma \ref{lem:rounding_prob} implies a lower bound on the expected objective value of the solution returned.
\begin{lemma}\label{lem:expected_value}
Algorithm~\ref{alg:rounding} returns an independent set $T \in \I$ with expected objective value $$\E\left[\det\pr{\sum_{i \in T} \bv_i \bv_i^\top}\right]\geq \gamma \sum_{S\in \I_d} \det\pr{\sum_{i\in S} x_i \bv_i \bv_i^\top}.$$
\end{lemma}

Next, we relate this lower bound to the objective of the convex relaxation \CPD{} in a two-step procedure. Building on results by~\cite{AnariGV18}, the lower bound on the expected objective of the algorithm can be bounded in terms of objective of a different convex relaxation as described in Lemma~\ref{lem:relate}. Proof of the lemma is inferred from the inequality proved in \cite{AnariGV18} on completely log-concave polynomials by observing that the polynomials (in  $\by$ and $\bz$ variables)
$\det\pr{\sum_{i=1}^n x_i^\star y_i \bv_i \bv_i^\top}$ and $\sum_{S \in \I_d} z^{[n]\setminus S} $ are completely log-concave. Here we use the notation \(\pr{\frac{\by\bw}{\balpha}}^{\balpha}:=\prod_{i=1}^n \pr{\frac{y_iw_i}{\alpha_i}}^{\alpha_i}\).  

\begin{lemma}\label{lem:relate}
For any $\bx^\star\geq 0$,
\[ \sum_{S \in \I_d} \det\pr{\sum_{i\in S} x_i^\star \bv_i \bv_i^\top} \geq e^{-2d} \sup_{\balpha \in P(\I_d)} \inf_{\by,\bw>0} \frac{\det\pr{\sum_{i=1}^n x_i^\star y_i \bv_i \bv_i^\top} \pr{\sum_{S \in \I_d} w^S }}{\pr{\frac{\by\bw}{\balpha}}^{\balpha}}.\]
\end{lemma}

To finish the proof of Theorem~\ref{thm:main2}, we show that the convex relaxation \CPD{} is stronger than the convex relaxation studied in \cite{AnariGV18}. 
\begin{lemma}\label{lem:stronger_cp}
For any $\bx^\star\geq 0$,
\[ \sup_{\balpha \in P(\I_d)} \inf_{\by,\bw>0} \frac{\det\pr{\sum_{i=1}^n x_i^\star y_i \bv_i \bv_i^\top} \pr{\sum_{S \in \I_d} w^S }}{\pr{\frac{\by\bw}{\balpha}}^{\balpha}} \geq \inf_{\bz\in \Z} \det\pr{\sum_{i=1}^n x_i^\star e^{z_i} \bv_i \bv_i^\top}.\]
\end{lemma}

Note that we cannot directly use the convex relaxation of~\cite{AnariGV18} and avoid the two-step procedure for our problem. Algorithm~\ref{alg:rounding} and the proof of Lemma~\ref{lem:rounding_prob} require that the solution $\bx$ is sparse, and we do not know if such property holds true for the convex relaxation of~\cite{AnariGV18}. 

Before we prove these lemmas, we use them to prove the main result of our paper. 

\begin{proofof}{Theorem~\ref{thm:main2}} 
We first show  \eqref{eq:thm-OPT-CP}. Recall that $ f(\bx) = \inf_{\bz\in\Z}\log\det\pr{\sum_{i \in [n]} x_ie^{z_i} \bv_i \bv_i^\top}$. Let $(\bx^\star,\bz^\star)$ be an optimal solution to \CPD, so we have $f(\bx^\star) = \CP{}$. The first inequality of \eqref{eq:thm-OPT-CP}  follows from Lemma~\ref{lem:alg}. It remains to show the second inequality.

By Theorem~\ref{thm:sparsity}, there exists $\hat{\bx} \in \P(\M)$ such that $f(\hat{\bx}) = f(\bx^\star)$ and $|\{ i \in [n] \mid 0<\hat {x_i} < 1\}| \leq 2\pr{{d+1 \choose 2}  + d}$. Let $T \in \I$ be the random solution returned by Algorithm~\ref{alg:rounding} given an input $\bx = \hat {\bx}$. We apply Lemmas~\ref{lem:expected_value},
\ref{lem:relate}, and~\ref{lem:stronger_cp} successively and in this order
to get
\begin{align}
\E\left[\det\pr{\sum_{i \in T} \bv_i \bv_i^\top}\right] &\geq (2e^3d)^{-d} \sum_{S\in \I_d} \det\pr{\sum_{i\in S} \hat{x_i} \bv_i \bv_i^\top} \nonumber \\
&\geq (2e^3d)^{-d} e^{-2d}  \sup_{\bal \in P(\I_d)} \inf_{\by,\bw>0} \frac{\det\pr{\sum_{i=1}^n \hat{x_i} y_i \bv_i \bv_i^\top} \pr{\sum_{S \in \I_d} w^S }}{\pr{\frac{\by\bw}{\balpha}}^{\bal}} \nonumber  \\
&\geq (2e^5d)^{-d} \inf_{\bz\in \Z} \det\pr{\sum_{i=1}^n \hat{x_i} e^{z_i} \bv_i \bv_i^\top} \nonumber \\
&= (2e^5d)^{-d} \inf_{\bz\in \Z} \det\pr{\sum_{i=1}^n x_i^\star e^{z_i} \bv_i \bv_i^\top} =(2e^5d)^{-d} \cdot \CP \label{eq:alg-vs-CP}
\end{align}
where the  first of the two equalities follows from Theorem~\ref{thm:sparsity}.

On the other hand, for any \(T\in\I\), we have \(\det\pr{\sum_{i \in T} \bv_i \bv_i^\top} \leq \OPT\), and therefore 
\begin{equation}
\E\left[\det\pr{\sum_{i \in T} \bv_i \bv_i^\top}\right] \leq \OPT{}. \label{eq:alg-vs-OPT}
\end{equation}
 Combining \eqref{eq:alg-vs-CP} and \eqref{eq:alg-vs-OPT} proves the second inequality of \eqref{eq:thm-OPT-CP}.

Given a solution $\bx^\star$ and $\bz^\star$ attaining the infimum in $\inf_{\bz\in \Z} \det\pr{\sum_{i=1}^n x_i^\star e^{z_i} \bv_i \bv_i^\top}$, the efficiency of the randomized algorithm that satisfies \eqref{eq:alg-vs-CP} follows from the efficiency of obtaining a sparse solution
(by Theorem~\ref{thm:sparsity}) and of the rounding Algorithm~\ref{alg:rounding}.
\end{proofof}


Now we prove Lemmas~\ref{lem:rounding_prob} and \ref{lem:expected_value}. The proofs of Lemmas  \ref{lem:relate} and~\ref{lem:stronger_cp} appears in Appendix \ref{sec:app-rounding}.

\begin{proofof}{Lemma~\ref{lem:rounding_prob}} We need to prove that for any $S \subseteq R_1 \cup R_2$ such that $S \in \I_d$, 
\[\Pr[ S \subseteq T] \geq (2e^3d)^{-d}.\]
Let $S_1 = S \cap R_1$ and $ S_2 = S \cap R_2$. Since $x_i =1$ for any $i \in R_2$ and $\bx \in \P(\M)$, we have $R_2 \in \I$. Since $S \in \I$ and $S_1 \subseteq S$, we have $S_1 \in \I$. We first claim the following.

\begin{claim}\label{claim:matroid_be}
There exists $Y \subseteq R_2 \setminus S_2$ such that $|Y| \leq |S_1|$ and $S_1 \cup (R_2 \setminus Y) \in \I$.
\end{claim}
\begin{proof}
Recall that $S = S_1 \cup S_2 \in \I$ and $R_2 \in \I$. If $|R_2| \leq |S_1\cup S_2|$, then $Y = R_2$ satisfies the condition. Else, by the definition of matroids, there exists an element $i \in R_2 \setminus (S_1 \cup S_2)$ such that $S_1 \cup S_2 \cup \{i \} \in \I$. Since $S_1 \cap R_2 = \emptyset$, we get that $i \in R_2 \setminus S_2$. 
Repeating this process for $|R_2| - |S_1 \cup S_2|$ times, we obtain a set $W\subseteq R_2 \setminus S_2$ of size $|R_2| - |S_1 \cup S_2|$ such that $S_1 \cup S_2 \cup W \in \I$. If $Y = R_2 \setminus (S_2 \cup W)$, then $ S_1 \cup (R_2 \setminus Y) = S_1 \cup S_2 \cup W  \in \I$.  Since $W$ has size $|R_2| - |S_1 \cup S_2|$, $Y$ has size $|R_2| - (|R_2| - |S_1 \cup S_2|) - |S_2| = |S_1\cup S_2| - |S_2| \leq |S_1|$.
\end{proof}

Let $Y \subseteq R_2 \setminus S_2$ be a set such that $S_1 \cup (R_2 \setminus Y) \in \I$. Next, we prove a lower bound on $\Pr[S \subseteq T]$. Note that $S$ is a disjoint union of $S_1$ and $S_2$. Hence,\begin{align*}
\Pr[S \subseteq T] & = \Pr[S_1 \subseteq T \text{ and } S_2 \subseteq T]\\
& \geq \Pr[T \cap R_1 = S_1  \text{ and } S_2 \subseteq T]\\
& \geq \Pr[T \cap Y = \emptyset \text{ and }T \cap R_1 = S_1
\text{ and } S_2 \subseteq T]\\
& = \Pr[ T \cap Y = \emptyset] \cdot \Pr[T \cap R_1 = S_1 \mid T \cap Y = \emptyset] \cdot \Pr[ S_2 \subseteq T \mid T \cap Y = \emptyset, T \cap R_1 = S_1]
\end{align*}

Next, we lower bound each of the probabilities.

\begin{enumerate}
\item \para{$\Pr[T \cap Y = \emptyset]$:} Consider the event that $T \cap Y = \emptyset$. It happens if for each $i \in Y$, $i$ is not added to $T$ during the execution of the algorithm. Let $T'$ be the set $T$ before the iteration considering $i$. If $T' \cup \{i\} \not \in \I$, $i$ is not added to $T$ with probability $1$. If $T' \cup \{i \} \in \I$, $i$ is not added to $T$ with probability $1/2$. Hence, for each $i \in T$, $i$ is not added to $T$ with probability at least $1/2$. Probability that none of the elements of $Y$ are added to $T$  is therefore at least $\pr{\frac{1}{2}}^{|Y|}$. Since $|Y| \leq |S_1|$,
\[ \Pr[Y \cap T = \emptyset] \geq \pr{\frac{1}{2}}^{|S_1|}.\]

\item \para{$\Pr[T \cap R_1 = S_1 \mid T \cap Y = \emptyset]$:} Since all elements of $R_1$ are considered before the elements of $R_2$ (and hence $Y$), we have 
\begin{equation} \label{eq:prob2-1}
\Pr[T \cap R_1 = S_1 \mid T \cap Y = \emptyset] = \Pr[T \cap R_1 = S_1].
\end{equation}
To get \(T \cap R_1 = S_1 \), we must have $(R_1 \setminus S_1) \cap T = \emptyset$ and $S_1 \subseteq T$. Hence,
\begin{equation} \label{eq:prob2-2}
 \Pr[T \cap R_1 = S_1] = \Pr[T \cap (R_1 \setminus S_1) = \emptyset]\cdot\
\Pr[S_1 \subset T \mid T \cap (R_1 \setminus S_1) = \emptyset].
\end{equation}

As argued above, for any element $i\in  R_1 \setminus S_1$, the probability that $i$ is not in $T$ (regardless of other elements) is at least $1-\frac{1}{d}$. Hence, $\Pr[T \cap (R_1 \setminus S_1)] \geq \pr{1-\frac{1}{d}}^{|R_1\setminus S_1|}$ which is equal to $\pr{1-\frac{1}{d}}^{|R_1|-|S_1|} $ since $S_1 \subseteq R_1$.

Since $S\in\I$ and $S_1 \subseteq S$, we have $S_1\in\I$. Consider an element $i \in S_1$ and the set $T'$ being the set $T$ before the algorithm processes the element $i$. If no element of $R_1 \setminus S_1$ is picked, then $T' \cup \{i\} \subset S_1$. Hence, $T' \cup \{i\} \in \I$, and the probability that the element $i$ is picked is $\frac{1}{d}$. Hence, if no element of $R_1 \setminus S_1$ is picked, then every element of $S_1$ is picked with probability $\frac{1}{d}$. This implies that
\begin{equation} \label{eq:prob2-3}
\Pr[S_1 \subseteq T \mid T \cap (R_1 \setminus S_1) = \emptyset] = \pr{\frac{1}{d}}^{|S_1|}.
\end{equation}

Combining \eqref{eq:prob2-1}-\eqref{eq:prob2-3}, we get
\[\Pr[T \cap R_1 = S_1 \mid T \cap Y = \emptyset] \geq \pr{1-\frac{1}{d}}^{|R_1|-|S_1|} \pr{\frac{1}{d}}^{|S_1|}.\]

\item \para{$\Pr[ S_2 \subseteq T \mid T \cap Y = \emptyset, T \cap R_1 = S_1]$:} Consider an element $i \in S_2$. Let $T'$ be the set $T$ just before the algorithm considers the element $i$. If $T' \cap R_1 = S_1$ and $T' \cap Y = \emptyset$, then $T' \cup \{i\} \subseteq S_1 \cup (R_2 \setminus Y)$. By Claim~\ref{claim:matroid_be}, $S_1 \cup (R_2 \setminus Y) \in \I$. Hence, if $T' \cap R_1 = S_1$ and $T' \cap Y = \emptyset$. Then, $T' \cup \{i\} \in \I$, and 
$i$ is added to $T$ with probability $1/2$. Therefore,
\[ \Pr[ S_2 \subseteq T \mid T \cap Y = \emptyset, T \cap R_1 = S_1] = \pr{\frac{1}{2}}^{|S_2|}.\]
\end{enumerate}

Combining the bounds on the three probabilities, we get
\[ \Pr[S \subseteq T] \geq \pr{\frac{1}{2}}^{|S_1|}\pr{1-\frac{1}{d}}^{|R_1|-|S_1|} \pr{\frac{1}{d}}^{|S_1|}\pr{\frac{1}{2}}^{|S_2|}.\]

Since $|S| = d$ and $S$ is a disjoint union of $S_1$ and $S_2$, we have $|S_1| + |S_2| = d$. Also, by the assumption of the theorem, $|R_1| \leq 2\pr{{d+1\choose 2} +d}$. Hence,
\[ \Pr[S \subseteq T] \geq \pr{\frac{1}{2}}^{d} \pr{1-\frac{1}{d}}^{2\pr{{d+1\choose 2} +d}}  \pr{1-\frac{1}{d}}^{-|S_1|}\pr{\frac{1}{d}}^{|S_1|}.\]
Since $|S_1| \leq d$,
we have\begin{align*}
 \Pr[S \subseteq T] &\geq  \pr{\frac{1}{2}}^{d} \pr{1-\frac{1}{d}}^{2\pr{{d+1\choose 2} +d}}\pr{1-\frac{1}{d}}^{-d}\pr{\frac{1}{d}}^{d}\\
& = \pr{\frac{1}{2d}}^{d} \pr{1-\frac{1}{d}}^{d(d+1) + d}
\end{align*}
For $d \geq 2$, we have $1-\frac{1}{d} \geq e^{- \frac{1.5}{d}} \geq e^{-\frac{3}{d+2}}$.
Hence,\begin{align*}
 \Pr[S \subseteq T] & \geq \pr{2d}^{-d} e^{-3d } =\pr{2e^3 d}^{-d}
\end{align*}
finishing the proof of Lemma~\ref{lem:rounding_prob}
\end{proofof}

\begin{proofof}{Lemma~\ref{lem:expected_value}} By Lemma~\ref{lem:rounding_prob}, for any $S \subseteq R_1 \cup R_2$ such that $S \in \I_d$, we have
$\Pr[S \subset T] \geq \pr{2e^3 d}^{-d}$. The rounding Algorithm \ref{alg:rounding} returns a solution \(T\) of expected value 
\begin{align*}
\E \left[ \det\pr{\sum_{i\in T} \bv_i \bv_i^\top}\right] &= \E \left[ \sum_{S \in {T \choose d}} \det\pr{\sum_{i \in S} \bv_i \bv_i^\top}\right]\\
& = \sum_{S \subseteq [n]: |S| = d} \Pr\left[S \subseteq T\right] \det\pr{\sum_{i \in S} \bv_i \bv_i^\top}
\end{align*}
where we apply the Cauchy-Binet formula to obtain the first equality. Since we only pick elements of $R_1 \cup R_2$ which form an independent set, we have
\begin{align*}
\E \left[ \det\pr{\sum_{i\in T} \bv_i \bv_i^\top}\right] & = \sum_{S \subseteq R_1 \cup R_2: S \in \I_d} \Pr\left[S \subseteq T\right] \det\pr{\sum_{i \in S} \bv_i \bv_i^\top}\\
& \geq \pr{2e^3d}^{-d} \sum_{S \subseteq R_1 \cup R_2: S \in \I_d} \det\pr{\sum_{i \in S} \bv_i \bv_i^\top}.
\end{align*}
For each $i \in [n]$, we have $0 \leq x_i \leq 1$. Hence,
\[ \E \left[ \det\pr{\sum_{i\in T} \bv_i \bv_i^\top}\right] \geq \pr{2e^3d}^{-d} \sum_{S \subseteq R_1 \cup R_2: S \in \I_d} \det\pr{\sum_{i \in S} x_i \bv_i \bv_i^\top}.\]
For $S \in \I_d$ such that $S \not\subseteq R_1 \cup R_2$, there exists $i \in S$ such that $x_i = 0$. Hence, for $S \in \I_d$ such that $S \not \subseteq R_1 \cup R_2$, $\sum_{i \in S} x_i \bv_i \bv_i^\top$ has rank at most $d-1$ and $\det\pr{\sum_{i \in S} x_i \bv_i \bv_i^\top} =0$. Therefore,

\[ \E \left[ \det\pr{\sum_{i\in T} \bv_i \bv_i^\top}\right] \geq \pr{2e^3d}^{-d} \sum_{ S \in \I_d} \det\pr{\sum_{i \in S} x_i \bv_i \bv_i^\top}\]
finishing the proof of Lemma~\ref{lem:expected_value}.
\end{proofof}

%% file: deterministic.tex
\section{Deterministic Algorithm}\label{sec:deterministic}

In this section, we prove Theorem~\ref{thm:deterministic} and give the
deterministic algorithm achieving the claimed guarantee. The algorithm
reduces the ground set in each iteration until the ground set is
itself an independent set. Given any $V\subseteq [n]$, we let
$\M_{|V}=(V,\I_{|V})$ denote the matroid obtained by deleting all
elements not in $V$ from $\M$. Moreover, we let \CPD{}(V) denote the
convex program when the ground set and the matroid are $V$ and
$\M_{|V}$, respectively, and we consider only vectors indexed by $V$.
We let \CP{}(V) denote optimal value of the convex program
\CPD{}(V). We denote by \(r(V)\) the rank of the matroid \(\M_{|V}\).

We first describe the deterministic rounding algorithm, presented in Algorithm \ref{alg:deterministic}.
\begin{algorithm}
\caption{Deterministic Algorithm}
\begin{algorithmic}[1]
\State \textbf{Input:} a matroid \(\M=([n],\I)\).
\State \textbf{Output:} a basis \(S\in\I\).
\Procedure{Rounding}{}
        \State Let \(\bx\) be optimal solution to \CPD{} such that \(|\{i\in [n]:0<x_i \}|\leq k+2\left(\binom{d+1}{2}+d\right) \) as returned by Theorem~\ref{thm:sparsity}.
        \State Let \(V\gets \set{i\in [n]:0<x_i}\).
        \While {\( V\notin \I\) }
            \State \(i\gets\argmax_{j\in V: r(V\setminus \{j\})=r(V)} \CP{}(V\setminus \{j\})\) (breaking a tie arbitrarily)
            \State \(V\gets V\setminus \{i\}\)
        \EndWhile
        \Return $V$
        \EndProcedure
\end{algorithmic}
\label{alg:deterministic}
\end{algorithm}

Observe that $V$ is initialized to a set of size at most
$k+2\left(\binom{d+1}{2}+d\right)$ along with $r(V)=k$. Moreover,
$\CP{}(V)=\CP{}$ initially, since we just remove all elements with $x_i=0$ from the ground set.

In each iteration of the while loop, we decrease the size of $V$ by one, and thus there can be at most $2\left(\binom{d+1}{2}+d\right)$ iterations of the while loop. In each iteration, we do not decrease the rank of \(V\) from $k$, so the final output, by construction, is  an independent set of size \(k\) and hence feasible. To prove the guarantee, we show that in each iteration,
\begin{equation}\label{eqn:decrease}
\CP{}(V\setminus \{i\})\geq \beta \cdot \CP{}(V)
\end{equation}
where \(\beta=(2e^5 d)^{-d}=O(d)^{-d}\). Also, the relaxation is exact after the
last iteration because $V$ is a basis after the while loop
terminates. Thus, the objective value of the returned solution is at
least
 \begin{equation*}
\beta^{2(\binom{d+1}{2}+d)} \cdot \CP{},
\end{equation*}
giving an approximation factor \(O(d)^{2d(\binom{d+1}{2}+d)}=O(d)^{d^3}\cdot O(1)^{3d^2\log d}=O(d)^{d^3}\), as claimed.

It only remains to prove \eqref{eqn:decrease}. From the guarantee of the randomized algorithm given in Theorem~\ref{thm:main2}, there \emph{exists} a basis $S\in \I_{|V}$ with $S\subseteq V$ such that
\begin{equation}\label{eq:apx-guarantee}
\det \left(\sum_{j\in S} \bv_j \bv_j^\top\right)\geq \beta\cdot \CP{}(V).
\end{equation}
Let $j\in V\setminus S$ where $j$ must exist since $V\notin \I$.  Then $r(V\setminus \{j\})=r(S)=k$ since $S$ is a basis.
We have
\(
\CP{}(V\setminus \{j\}) \ge \det \left(\sum_{e\in S} \bv_e
  \bv_e^\top\right),
\) because the indicator vector $\bx$ of
$S \subseteq V\setminus \{j\}$ is a solution to $\CPD{}(V\setminus \{j\})$ of value
$\det \left(\sum_{e\in S} \bv_e \bv_e^\top\right)$. Together with
\eqref{eq:apx-guarantee}, and because $i$ is chosen to maximize
$\CP{}(V\setminus \{j\})$ over $j$ s.t.~$r(V\setminus \{j\})=k$, we
have established \eqref{eqn:decrease}. This completes the proof of
Theorem~\ref{thm:deterministic}.

%% file: preliminaries.tex
\section{Preliminaries}\label{sec:prelim}

In this section, we introduce definitions and theorems from convex
duality and the polyhedral theory of matroids that we use in this paper.

\begin{theorem}[Cauchy-Binet Formula] For any set of vectors $\bv_1,\dots,\bv_n \in \rR^d$,
\[ \det\pr{\sum_{i=1}^n \bv_i \bv_i^\top} = \sum_{S \in {[n] \choose d}} \det\pr{\sum_{i \in S} \bv_i \bv_i^\top}\]
where ${[n] \choose d}$ denotes the set of subsets of $[n]$ of size $d$.
\end{theorem}

\subsection{Matroids}\label{sec:matroidprelim}
For basic preliminaries on matroids, we refer readers to Chapter 39 of \cite{schrijver2003combinatorial}. Here we include the definitions and basic facts that are used in our proofs in this paper.
\begin{definition} [Matroids]
A matroid $\M = (E,\I)$ is a structure consisting of a finite ground set $E$ and a non-empty collection $\I$ of independent subsets of $E$ satisfying the following conditions:
\begin{itemize}
\item If $S \subseteq T$ and $T \in \I$, then $S \in \I$.
\item If $S,T \in \I$ and $|T| > |S|$, then there exists an element $i \in T \setminus S$ such that $S \cup \{i\} \in \I$.
\end{itemize}
\end{definition}

\begin{definition}[Rank Function of Matroids]
For a matroid $\M = (E,\I)$, the rank of $\M$, denoted by \(r(\M)\), is the size of the largest independent set. For any \(S\subseteq E\), the rank function \(r:2^E\rightarrow\R\), denoted by \(r(S)\), is the size of maximal independent sets contained in \(S\).
\end{definition}
It is known that all maximal sets in a subset \(S\subseteq E\) have the same size, and therefore \(r(\M)\) is well-defined and equals to the size of these maximal sets.
\begin{definition} [Basis of Matroids]
An independent set with the largest cardinality of a matroid \(\M\), i.e. with the size equals to the rank of \(\M\), is called a \textit{basis} of $\M$. The set of all bases of \(\M\) is called a \textit{base set} of \(\M\) and is denoted by $\B_\M$.
\end{definition}
\begin{lemma} [Strong Basis Exchange Property] (Theorem 39.12 of \cite{schrijver2003combinatorial})
Let \(B_1,B_2\) be distinct bases of a matroid \(\M\). Then for all \(x\in B_2\setminus B_1\), there exists \(y\in B_1\setminus B_2\) such that \(B_1+x-y,B_2+y-x\in\B_\M\).
\end{lemma}
\begin{definition}[Matroid Base Polytope]
For a matroid $\M=(E,\I)$ and a subset \(S\subseteq E\) the indicator vector \(\one_S\) of \(S\) is a binary vector in $\{0,1\}^n$ whose $i^{th}$ coordinate is $1$ if $i\in S$ and $0$ otherwise. The matroid base polytope $\P(\M)$ is the convex hull of indicator vectors of all the bases of $\M$.
\end{definition}
\begin{lemma}  [Characterization of Base Polytope] (Corollary 40.2d of \cite{schrijver2003combinatorial})
For a matroid $\M$, the matroid base polytope $\P(\M)$ is characterized by
\begin{equation}
\P(\M)=\set{x\in\R^E:\sum_{e\in S} x_e \leq r(S),\ \forall S\subseteq [n] \text{ and } \sum_{e\in E} x_e =r(E)}.
\end{equation}
Therefore, for any matroid $\M = (E,\I)$ of rank $r$ and $\mathbf{x} \in \P(\M)$, we have $\sum_{e \in E} x_e = r$.
\end{lemma}
\begin{lemma}  [Submodularity of the Rank Function] (Theorem 39.8 of \cite{schrijver2003combinatorial})
For a matroid $\M=(E,\I)$, the rank function \(r\) of \(\M\) is submodular. That is, for all sets \(A,B\subseteq E\),
\begin{equation}
r(A)+r(B)\geq r(A\cup B)+r(A\cap B).
\end{equation}
\end{lemma}

\begin{definition} [Basis Generating Polynomial of Matroids]
For a matroid $\M = ([n],\I)$ with the base set $\mathcal{B}_\M $, the basis generating polynomial of matroid $\M$ is
\[ g_\M(z_1,\dots,z_n) = \sum_{B \in \mathcal{B_{M}}} \prod_{i \in B} z_i.\]
\end{definition}

\subsection{Optimality Conditions for Convex Programming}\label{sec:gen-KKT}

In this section, we recall Slater's constraint qualification as a sufficient
condition for strong duality to hold for a convex program.
Since all constraints in the optimization problems we consider in this
paper are affine, we will state the relevant results for this special
case.

Consider an optimization problem
\begin{equation}\label{eq:inf-generic}
    \inf\{h(\bz): \bz \in \D, \bA\bz \le \bb, \bC\bz = \bd\},
\end{equation}
where $h$ is a convex function defined on a non-empty convex domain $\D \subseteq
\R^n$, $\bA$ and $\bC$ are matrices of dimensions, respectively, $m
\times n$ and $\ell \times n$, and $\bb \in \R^m$, $\bd\in \R^\ell$ are
vectors. We allow either $m$ or $\ell$ to be $0$. Then, the Lagrangian
associated with \eqref{eq:inf-generic} is
\begin{equation*}
  L(\bz, \blambda) := h(\bz) + \sum_{i = 1}^m\lambda_i (\bA\bz - \bb)_i + \sum_{i
    = 1}^\ell{\lambda_{m + i} (\bC\bz- \bd)_i},
\end{equation*}
defined for $\bz \in \D$ and Lagrange multipliers $\blambda \in \EE$, where $\EE = \{\blambda
\in \R^{m + \ell}: \lambda_i \ge 0 \ \ \forall i = 1, \ldots, m\}$.
It is easy to see that if $\bz\in \D$ is feasible, i.e., satisfies $\bA\bz
\le \bb$ and $\bC\bz = \bd$, and if $\blambda \in \EE$, then $L(\bz,\blambda) \le
h(\bz)$. Moreover,
\[
\sup_{\blambda \in \EE} L(\bz, \blambda) =
\begin{cases}
  h(\bz) & \bA\bz \le \bb, \bC\bz = \bd\\
  \infty & \text{otherwise}
\end{cases}.
\]
Therefore, we have
\begin{align*}
\sup_{\blambda  \in \EE}\inf_{\bz\in \D} L(\bz,\blambda) &\le
\sup_{\blambda  \in \EE}\inf_{\bz\in \D} \{L(\bz,\blambda): \bA\bz \le \bb, \bC\bz = \bd\}\\
&\le
\inf\{h(\bz): \bz \in \D, \bA\bz \le \bb, \bC\bz = \bd\}
= \inf_{\bz \in \D} \sup_{\blambda \in \EE} L(\bz, \blambda).
\end{align*}
The next theorem, which is classical, gives a sufficient condition for
these inequalities to hold with equality.

\begin{theorem}[Strong Duality]\label{thm:slater}
  Suppose that $h$ in \eqref{eq:inf-generic} is convex with non-empty
  convex domain $\D$.  Suppose that the infimum in
  \eqref{eq:inf-generic} is not $-\infty$, and that there exists some
  $\bz$ in the relative interior of $\D$ such that $\bA\bz \le \bb$ and $\bC\bz =
  d$. Then there exists a vector $\blambda^\star \in \EE$ such that
\[
\inf_{\bz \in \D} L(\bz, \blambda^\star)
=
\inf \{h(\bz): \bz\in \D, \bA\bz \le \bb, \bC\bz = \bd\}.
\]
\end{theorem}

Furthermore, the same assumptions also imply that the KKT optimality
conditions are necessary and sufficient.

\begin{theorem}[KKT Conditions] \label{thm:gen-KKT} Suppose that the
  assumptions of Theorem~\ref{thm:slater} hold and that $\bz^\star \in
  \D$ is feasible, i.e., $\bA\bz^\star \le \bb$, $\bC\bz^\star = \bd$, and
  $\blambda^\star \in \EE$. Then, $\bz^\star$ achieves the infimum in
  \eqref{eq:inf-generic} and $\blambda^\star$ achieves
  \(
  \inf_{\bz \in \D} L(\bz, \blambda^\star)
  = h(\bz^\star)
  \)
  if and only if
  \begin{enumerate}[label={\alph*.}]
  \item $\lambda^\star_i (\bA\bz^\star - b)_i = 0$ for all $i \in \{1, \ldots, m\}$, and\label{cond:gen-CS}
  \item $0 \in \partial h(\bz^\star) + \bM^\top \blambda$,\label{cond:gen-first-order}
  \end{enumerate}
  where $\partial h(\bz^\star)$ is the subgradient of $f$ at $\bz^\star$ and
  $\bM$ is the matrix \(\bM = \begin{pmatrix}\bC \\ D\end{pmatrix}\).
\end{theorem}
For proofs of Theorem~\ref{thm:slater} and Theorem~\ref{thm:gen-KKT},
see e.g.~Theorems~28.2~and~28.3~in~\cite{Rockafellar}.


%% file: log-concavity.tex
\subsection{Complete Log-Concavity}



In this section, we review some of the results by Anari, Gharan, and Vinzant~\cite{AnariGV18} on log-concave polynomials and their implications for our problem.
For vectors \(\by,\bp\in\R^n\), we let \(\by^\bp:=\prod_{i=1}^n y_i^{p_i}\) and \(\by^{1-\bp}:=\prod_{i=1}^n y_i^{1-p_i}\).
We define a vector multiplied and divided by another vector or a
scalar coordinate-wise.
For a vector \(\bal\in\R^n\), we also denote \(|\bal|:=\sum_{i=1}^n \alpha_i\).
\begin{definition} [Log-Concave Polynomials]
A polynomial $g \in \rR[z_1,\dots,z_n]$ with non-negative coefficients is log-concave if $\log(g)$ is concave  over $\rR_{>0}^n$. Equivalently, $g$ is log-concave if for any two vectors $\bv,\bw \in \rR_{\geq 0}^n$ and $\lambda \in [0,1]$, we have
\[ g(\lambda \mathbf{v} + (1-\lambda) \bw) \geq g(\bv)^\lambda \cdot g(\bw)^{1-\lambda}.\]
\end{definition}

\begin{lemma}\label{lem:log_concavity_preserve}\label{lem:prod_log_concave}(Proposition 2.2 in \cite{AnariGV18})
\begin{itemize}
\item For any two log-concave polynomials $g,h$, the polynomial $g\cdot h$ is log-concave.
\item For any log-concave polynomial $g(z_1,\dots,z_n)$, the polynomial
$c\cdot g(\lambda_1z_1,\dots,\lambda_nz_n)$ is log-concave if $c,\lambda_1,\dots,\lambda_n \geq 0$.
\end{itemize}
\end{lemma}

\begin{definition} [Completely Log-Concave Polynomial]
A polynomial $g \in \rR[z_1,\dots,z_n]$ is completely log-concave if for every $k \geq 0$ and nonnegative matrix $\mathbf{V} \in \rR_{\geq 0}^{n \times k}$, $D_{\mathbf{V}}g(\mathbf{z})$ is nonnegative and log-concave as a function over $\rR^n_>0$, where
\[ D_{\mathbf{V}}g(\mathbf{z}) = \pr{ \Pi_{j=1}^k \sum_{i=1}^n V_{ij} {\partial}_i} g(\mathbf{z}).\]
\end{definition}

\begin{lemma}\label{lem:matroid_log_concave}(Theorem 4.2 in ~\cite{AnariGV18}) For any matroid $\M$, the basis generating polynomial $g_\M(\mathbf{z})$ is completely log-concave over the positive orthant.
\end{lemma}


\begin{lemma}\label{lem:AGV_bound}(Corollary 7.2 in \cite{AnariGV18})
For any completely log-concave multi-affine polynomial $g \in \rR[y_1,\dots,y_n, z_1,\dots,z_n]$ and $\mathbf{p} \in [0,1]^n$, the following inequality holds:

\[ \pr{\Pi_{i=1}^n ({\partial}_{y_i} + {\partial}_{z_i})} g(\mathbf{y},\mathbf{z}) |_{\mathbf{y} = \mathbf{z} = 0} \geq \pr{\frac{\mathbf{p}}{e^2}}^{\mathbf{p}} \inf_{\by,\bz \in \rR_{>0}^n} \frac{g(\by,\bz)}{\by^{\mathbf{p}} \bz^{1-\mathbf{p}}} .\]
\end{lemma}

\begin{lemma}\label{lem:det-is-log-concave}(Corollary 1.8 in~\cite{AnariLGV19})
For any set of vectors $\bv_1,\dots,\bv_n \in \rR^d$, the polynomial $\det(\sum_{i=1}^n x_i \bv_i \bv_i^\top)$ is a completely log-concave polynomial in $\bx$.
\end{lemma}
%

\begin{lemma}\label{lem:prod_complete_log_concave}(Follows from Theorem 5.3 and Corollary 5.5 in~\cite{branden2019lorentzian})
For any two completely log-concave homogenous polynomials $g \in \rR[y_1,\dots,y_n]$ and $h \in \rR[z_1,\dots,z_m]$, $g\cdot h$ is completely log-concave.
\end{lemma}

\begin{lemma}\label{lem:max_entropy_dist}(Implied by Theorem 2.10 in~\cite{AnariGV18})
Let $\zeta = \{S_1,\dots,S_t \subseteq [n]\}$ and $\P$ be the convex closure of $1_{S_1},\dots,1_{S_t}$. Then, for any point $\mathbf{p}$ strictly inside $\P$, there exists $\lambda_1,\dots,\lambda_n>0$ and a distribution $\mu$ over $S_i$'s such that $\mu(S_i) \propto \lambda^{S_i}$ and $p_i = \Pr_{S \sim \mu}[i \in S]$.
\end{lemma}

For a distribution $\mu: 2^{[n]} \rightarrow \rR_+$, the generating polynomial of $\mu$ is defined as $g_{\mu}(\bz) = \sum_{S \subseteq [n]} \mu(S) \Pi_{i \in S} z_i$. We also call \(\Pr_{S \sim \mu}[i \in S]\) the \textit{marginal probability} of an element \(i\) of the distribution \(\mu\). A distribution $\mu$ is called {\em log-concave} if the generating polynomial of $\mu$ is log-concave.

\begin{lemma}\label{lem:compare_entropy}(Theorem 5.2 in~\cite{AnariGV18}) For any log-concave distribution $\mu:2^{[n]} \rightarrow \rR_+$ with marginal probabilities $\mu_1,\dots,\mu_n\geq 0$, we have

\[ \mathcal{H}(\mu) := \sum_{S \subseteq[n]} \mu(S) \log \frac{1}{\mu(S)} \geq \sum_{i =1}^n \mu_i \log \frac{1}{\mu_i}.\]

\end{lemma}

%% file: app-relax.tex
\subsection{Convex Relaxation}
Here, we show that the convex program \eqref{eq:D-relax} is a
relaxation of \DetMax.

\begin{lemma} \label{lem:relaxation}
The optimization \eqref{eq:D-relax}:
\[
\sup_{\bx\in\P(\M)}  \inf_{\bz\in\Z} g(\bx,\bz) := \log\det\pr{\sum_{i \in [n]} x_ie^{z_i} \bv_i \bv_i^\top}
\]
is a relaxation of \DetMax{} problem \eqref{eq:D-obj}:
\[
\max\left\{\det\pr{\sum_{i\in S} \bv_i \bv_i^\top}: S \in \B\right\}.
\]
More specifically, \(\OPT\leq\exp(\CP)\).

\end{lemma}
\begin{proof}
Let \(S^\star\subseteq [n]\) denote an optimal set for  \DetMax{} and \(\bx^\star\) denote the indicator vector of \(S^\star\). We have
\begin{align*}
\CP = \sup_{\bx\in\P(\M)}  \inf_{\bz\in\Z} g(\bx,\bz) &\geq \inf_{\bz\in\Z} g(\bx^\star,\bz).
\end{align*}
For each \(\bz\in\Z\), we have
\begin{align*}
 \exp(g(\bx^\star,\bz))=\det\pr{\sum_{i=1}^n x_i^\star e^{z_i} \bv_i \bv_i^\top} &= \det\pr{\sum_{i\in{S^\star}} e^{z_i} \bv_i \bv_i^\top}=\sum_{R\subseteq S^\star:|R|=d} \det\pr{\sum_{i\in R}  e^{z_i} \bv_i \bv_i^\top}
\end{align*}
where we use the Cauchy-Binet formula for the last equality. For each \(R\subseteq S^\star\) of size \(d\), we have
\[
\det\pr{\sum_{i\in R} e^{z_i} \bv_i \bv_i^\top} =\pr{ \prod_{i\in R} e^{z_i}} \det\pr{\sum_{i\in R}  \bv_i \bv_i^\top} \geq \det\pr{\sum_{i\in R}  \bv_i \bv_i^\top}
\]
where the last inequality follows from the constraint \(z(S)\geq0,\forall S\in\I_d(\M)\) in the definition of \(\Z\) and \(R\in\I_d(\M)\). Therefore, we obtain
\begin{equation} \label{eq:last-relax}
\exp(g(\bx^\star,\bz))\geq\sum_{R\subseteq S^\star:|R|=d}\det\pr{\sum_{i\in R}  \bv_i \bv_i^\top} = \det\pr{\sum_{i\in{S^\star}}  \bv_i \bv_i^\top} =\OPT
\end{equation}
where we apply the Cauchy-Binet formula again for the first equality. Since \eqref{eq:last-relax} holds for each \(z\in\Z\), we have \(\exp(\inf_{\bz\in\Z} g(\bx^\star,\bz)) \geq \OPT\), and therefore \(\exp(\CP) \geq \exp(\inf_{\bz\in\Z} g(\bx^\star,\bz))\geq\OPT\).
\end{proof}

%% file: reduction.tex
\section{Preprocessing and Solvability} \label{sec:preprocessing}

In this section, we show how to transform the matroid and the input
vectors so that Lemma~\ref{lem:alg} holds. We first prove Lemma~\ref{lem:finite},
which gives sufficient conditions for the $\inf_{\bz \in
  \Z}g(\bx,\bz)$ to attain its infimum. This motivates the modifications to the
input, which we carry out next. Finally, we argue that the
appropriately modified input gives an efficiently solvable convex
relaxation.

\subsection{Attaining the Infimum}

In this section, we prove Lemma~\ref{lem:finite}.
We first prove an auxiliary lemma, from which the result will
follow. For notational convenience, let us define vectors $\hat{\bv}_i =
x_i \bv_i$, and let us denote by $\V$ the bases of the linear
matroid generated by the $\set{\hat{\bv}_i}_{i=1}^n$, i.e.
\[\V = \left\{S \subseteq[n]: |S| = d \text{ and } \det\left(\sum_{i \in S}\hat{\bv}_i \hat{\bv}_i^\top\right) \neq 0\right\}.\]

Recall that $\P(\V)$ is the convex hull of indicator vectors of sets
in $\V$. We denote by $\mathrm{relint}\,\P(\V)$ the relative interior
of $\P(\V)$. Equivalently, $\mathrm{relint}\,P(\V)$ is the set of all
points that can be written as a convex combination of indicator
vectors of $\V$ such that all coefficients in the convex combination
are positive. We claim the following lemma.



\begin{lemma} \label{lm:finite-inf}
Suppose that \(\P(\I_d(\M)) \cap \mathrm{relint}\, \P(\V) \neq
\emptyset.\) Then, \(\inf_{\bz \in \Z} g(\bx,\bz)\) is achieved
at some \(\bz^\star\in\Z\),
and $-g(\bx, \bz^\star)$ is equal to
 \begin{equation}\label{eq:dual}
   \inf_{\bm{\mu}\in \D, \bm{\nu} \in \R^{\I_d(\M)}}
    \left\{\sum_{S \in \V} \mu_S \log\left(\frac{\mu_S}{c_S}\right):
        \sum_{I \in \I_d(\M)} \nu_I 1_I = \sum_{S \in \V} \mu_S 1_S,
        \sum_{S \in \V} \mu_S = 1,
        \bm{\nu} \ge 0
        \right\},
\end{equation}
where $c_S = \det\left(\sum_{e \in S}\hat{\bv}_i \hat{\bv}_i^\top\right)$ and \(\D = \{\bm{\mu} \in \R^{\V}: \mu_S > 0 \ \forall S \in \V\}.\)
\end{lemma}
\begin{proof}
The objective function
$h(\bm{\mu}) := \sum_{S \in \V} \mu_S\log\left(\frac{\mu_S}{c_S}\right)$ is
easily seen to be convex in $\bm{\mu}$, and $\D$ is a non-empty convex
set. Moreover, since $\D$ is open, it is equal to its relative
interior. Let $\bx' \in \P(\I_d(\M)) \cap \mathrm{relint}\, \P(\V)$. Then, we can write
\begin{align*}
    \bx' &= \sum_{I \in \I_d(\M)} \nu_I 1_I
      = \sum_{S \in \V} \mu_S 1_S
\end{align*}
for some $\bm{\nu} \in \R^{\I_d(\M)}_{\ge 0}$ such that $\sum_{I \in \I_d(\M)} \nu_I = 1$
and $\bm{\mu} \in \D$ such that $\sum_{S \in \V}\mu_S = 1$. Therefore,
\eqref{eq:dual} is feasible, and the assumptions of
Theorem~\ref{thm:slater} are satisfied. We  claim that $-\inf_{\bz
  \in \Z}{g(\bx,\bz)} = \sup_{\bz \in \Z}{-g(\bx, \bz)}$ is
equivalent to the dual problem of \eqref{eq:dual}.

We can write the Lagrangian of \eqref{eq:dual} as
\begin{multline*}
L(\bm{\mu}, \bm{\nu}, \bz,\bm{\gamma}, t)
=
\sum_{S \in \V} \mu_S \log\left(\frac{\mu_S}{c_S}\right)
+ \sum_{i=1}^n z_i \left(\sum_{I \in \I_d(\M): i \in I} \nu_I - \sum_{S \in \V: i\in S} \mu_S\right) \\
+ t\left(\sum_{S\in \V}\mu_S - 1\right) - \sum_{I \in \I_d(\M)} \gamma_I \nu_I,
\end{multline*}
where the Lagrange multipliers are $\blambda = (\bz, \bm{\gamma}, t)$, and we
have $\bm{\gamma} \in \R_{\ge 0}^{\I_d(\M)}$, while $t$ and $\bz$ are
unconstrained. We will show that $\sup_{\bz \in \Z}{\,-g(\bx, \bz)}$ is
equivalent to
\begin{equation}\label{eq:lagr-dual}
\sup_{\bz\in \R^n,t\in \R,\bm{\gamma} \in \R_{\ge 0}^{\I_d(\M)}}\ \ \inf_{\bm{\mu}\in\D,\bm{\nu} \in \R^{\I_d(\M)}} L(\bm{\mu}, \bm{\nu}, \bz,\bm{\gamma}, t).
\end{equation}
In particular, we will show that for any $\bz$,
\begin{equation}\label{eq:duality-equiv}
\sup_{t\in \R,\bm{\gamma} \in \R_{\ge 0}^{\I_d(\M)}}\ \ \inf_{\bm{\mu}\in\D,\bm{\nu} \in \R^{\I_d(\M)}} L(\bm{\mu}, \bm{\nu}, \bz,\bm{\gamma}, t) =
\begin{cases}
  -g(\bx,\bz) &\bz\in \Z,\\
  -\infty &\bz \not \in \Z
\end{cases}.
\end{equation}
Since, by Theorem~\ref{thm:slater}, the supremum in
\eqref{eq:lagr-dual} is achieved and equals \eqref{eq:dual}, the lemma will follow.

Let us fix some $\bm{\mu}, \bz, \bm{\gamma}, t$, and first take the infimum over
$\bm{\nu}$. The terms in $L(\bm{\mu},\bm{\nu},\bz,\bm{\gamma},t)$ that depend on $\bm{\nu}$ are
\[
\sum_{I \in \I_d(\M)}\nu_I (z(I)-\gamma_I).
\]
We see that
\[
\inf_{\bm{\nu} \in \R^{\I_d(\M)}}\sum_{I \in \I_d(\M)}\nu_I (\gamma_I - z(I)) =
\begin{cases}
0 & z(I) = \gamma_I\ \forall I\\
-\infty & \text{otherwise}
\end{cases}.
\]
Recall that $\Z = \{\bz: z(I) \ge 0\ \forall I \in \I_d(\M)\}$. If $\bz
\not \in \Z$ and $\bm{\gamma} \in \R_{\ge 0}^{\I_d(\M)}$, then
$z(I) \neq \gamma_I$ for some $I\in\I_d(\M)$, and we have
$\inf_{\bm{\nu} \in \R^{\I_d(\M)}} L(\bm{\mu}, \bm{\nu}, \bz,\bm{\gamma},t) = -\infty$. If $\bz \in \Z$, then $\inf_{\bm{\nu} \in \R^{\I_d(\M)}} L(\bm{\mu}, \bm{\nu}, \bz,\bm{\gamma},t) = -\infty$ unless $z(I) = \gamma_I$ for every $I \in \I_d(\M)$. So, we may restrict the domain of $\bz$ to $\Z$  and simplify $L(\bm{\mu}, \bm{\nu}, \bz, \bm{\gamma},t)$ to
\begin{align*}
L'(\bm{\mu}, \bz,t) &:=
\sum_{S \in \V} \mu_S \log\left(\frac{\mu_S}{c_S}\right)
- \sum_{i=1}^n z_i \sum_{S \in \V: i \in S} \mu_S  + t\left(\sum_{S\in \V}\mu_S - 1\right)\\
&=
\sum_{S \in \V} \mu_S \log\left(\frac{\mu_S}{c_S}\right)
+  \sum_{S \in \V} \mu_S(t - z(S)) - t.
\end{align*}
In other words, either $\inf_{\bm{\nu} \in \R^{\I_d(\M)}} L(\bm{\mu}, \bm{\nu}, \bz, \bm{\gamma},t) = -\infty$, or else we have $\bz \in \Z$ and
\[
\inf_{\bm{\nu} \in \R^{\I_d(\M)}} L(\bm{\mu}, \bm{\nu}, \bz, \bm{\gamma},t) = L'(\bm{\mu}, \bz,t).
\]

Let us next fix $t$ and $\bz$, and compute $\inf_{\bm{\mu} \in D} L'(\bm{\mu}, \bz,
t)$. By taking derivatives over $\bm{\mu}$, we see that the infimum is achieved for ${\mu_S} :={c_S} e^{z(S)-t-1}$ and is equal to
\[
\inf_{\bm{\mu} \in D} L'(\bm{\mu}, \bz, t)
=
-e^{-t-1}\sum_{S \in \V}c_S e^{z(S)} - t.
\]
Taking the derivative over $t$, we see that the right-hand side is
maximized for $t = \log\left(\sum_{S \in \V}c_S e^{z(S)}\right) - 1$, and we have, for every $\bz$,
\[
\sup_{t \in \R}\inf_{\bm{\mu} \in D} L'(\bm{\mu}, \bz, t)
=
-\log\left(\sum_{S \in \V}c_S e^{z(S)}\right)
= -\log \det\left(\sum_{i=1}^n\hat{\bv}_i \hat{\bv}_i^\top\right)
= -g(\bx,\bz),
\]
where the penultimate equality follows by the Cauchy-Binet
formula. This establishes \eqref{eq:duality-equiv} and proves the
lemma.
\end{proof}

We can now prove Lemma \ref{lem:finite}.

\begin{proofof}{Lemma \ref{lem:finite}}
  By Lemma~\ref{lm:finite-inf}, we only need to show that
 \[\P(\I_d(\M)) \cap \mathrm{relint}\, \P(\V) \neq \emptyset.\]
 
Since we assumed that the vectors are in general position, we
 have \(\V=\set{S\subseteq \supp{\bx}:|S|=d},\) and therefore
 \[
 \mathrm{relint}\, \P(\V) =
 \left\{\bx' \in \R^n: 0 < x'_i < 1 \ \forall i \in \supp{\bx} \text{ and }
x'_i = 0
 \ \forall i \not \in \supp{\bx} \text{ and } \sum_{i=1}^n x'_i = d\right\}.
 \]

We first claim the following.
\begin{claim} \label{clm:i-in-notin-S}
  For any \(\bx\in\P(\M)\) such  that $\max_{i \in [n]} x_i < 1$, we have that for all
  \(i\in\supp{\bx}\),
\begin{enumerate}
\item
there exists \(S\in\I_d(\M)\) such that \(S\subseteq \supp{\bx}\) and \(i\in S\), and
\item
there exists \(S'\in\I_d(\M)\) such that \(S'\subseteq \supp{\bx}\) and \(i\notin S\)
\end{enumerate}
\end{claim}
\begin{proof}
  Since \(\bx\in\P(\M)\), \(\bx\) is a convex combination of indicator
  vectors of sets \(\set{S_j}_{j\in J}\) in \(\I_d(\M)\). For each
  \(i\in\supp{\bx}\), we have \(0<x_i<1\). Hence, there exist
  \(S,S'\in \set{S_j}_{j\in J}\) such that \(i\in S\) and \(i\notin
  S'\). Since \(\bx\) is a convex combination of indicator vectors
  of \(\set{S_j}_{j\in J}\), we must have \(S,S'\subseteq \supp{\bx}\).
\end{proof}

We now construct a point $\bx'$ in \(\P(\I_d(\M)) \cap
\mathrm{relint}\, \P(\V)\).  For each \(i\in\supp{\bx}\), we apply
Claim \ref{clm:i-in-notin-S} to obtain \(S_i,S_i'\in\I_d(\M)\) such
that \(i\in S_i\) and \(i\notin S_i'\). Let \(\bx'\in\P(\I_d(\M))\) be the
average of all indicator vectors of
\(\set{S_i}_{i\in\supp{\bx}}\cup\set{S_i'}_{i\in\supp{\bx}}\). Then, we
have that \(0<x'_i<1\) for all \(i\in \supp{\bx}\),
$\supp{\bx'}=\supp{\bx}$, and $\sum_{i=1}^dx'_i = d$. Therefore, $\bx' \in
\mathrm{relint}\, \P(\V)$.
\end{proofof}

\subsection{Transforming the Input}

In order to guarantee that the assumptions of Lemma~\ref{lem:finite} hold,
we transform the input to our problem \DetMax. The transformation will
preserve the values of integral solutions. It  consists of two
steps: first we construct a new matroid and corresponding new
vectors, and then we perturb the vectors to ensure that they are in general
position.

We define the new matroid \(\M'=(\U',\I')\) to be derived
from the original matroid \(\M = ([n], \I)\) by introducing two copies
for each element, forming a circuit. In particular, we let
\(\U'=[n]\times [2]\) and
\begin{align}
\I'&:=\set{S\subseteq \U':\forall i\in[n], \set{(i,1),(i,2)}\nsubseteq S \text{ and } \set{j\subseteq[n]:(j,1)\in S \text{ or } (j,2)\in S}\in\I}. \label{eq:I'}
\end{align}
We also create new vectors, corresponding to the elements of $\U'$, by
setting $\bv'_{(i,j)} = \bv_i$ for any $i\in [n]$ and $j\in [2]$.

The next claim shows that if \(\M\) is a matroid, then so is
\(\M'\), and that the objective value of \DetMax{} on the new
instance is preserved.

\begin{claim}\label{cl:new-mat}
For any  \(\M=([n],\U)\), the set system \(\M'=(\U',\I')\) constructed
by \(\U'=[n]\times [2]\) and $\I'$ defined as in \eqref{eq:I'}
is a matroid. Moreover, the vectors $\{\bv'_{(i,j)}: (i,j) \in \U'\}$
defined by $\bv'_{i,j} = \bv_i$ for all $i$ and $j$ satisfy
\[
\max_{S \in \B} \det\left(\sum_{i \in S}{\bv_i \bv_i^\top}\right)
=
\max_{S' \in \B'} \det\left(\sum_{e \in S}{\bv'_e (\bv'_e)^\top}\right),
\]
where $\B'$ are the bases of $\M'$.
\end{claim}
\begin{proof}
We first show that $\M'$ is a matroid. If \(B\in\I'\) and \(A\subseteq
B\), then \(A\in \I'\) by the hereditary property of
\(\I\). Let \(A,B\in\I'\) be such that \(|B|>|A|\). Define  \(A_\M =
\set{j\subseteq[n]:(j,1)\in A \text{ or } (j,2)\in A}\) and similarly
for \(B_\M\). Then, by the definition of \(\I'\), we have
\(|A|=|A_\M|,|B|=|B_\M|\), and \(A_\M,B_\M\in\I\). Therefore, there
exists \(b\in B_\M\setminus A_\M\) such that \(A_\M+b\in\I\). Suppose
\((b,1)\in B\). Then, \((b,1)\in B\setminus A\) and \(A+(b,1)\in\I'\),
finishing the proof. The other case \((b,2)\in B \) is similar.

To show that the optimal value of \DetMax{} is preserved,
observe  that, for any $S \in \B$, the set $S' = \{(i,1): i \in S\}$
is a basis of $\M'$, and observe that
\[
\det\left(\sum_{i \in S}{\bv_i \bv_i^\top}\right)
=
\det\left(\sum_{e \in S'}{\bv'_e (\bv')_e^\top}\right).
\]
In the other direction, we have that, by the definition of $\M'$, for
any $S' \in \B$, the set $S = \{i: (i,1) \in S' \text{ or } (i,1) \in
S'\}$ is a basis of $\M$, and the equality above is, again,
satisfied.
\end{proof}

The next step is to transform the vectors so that they are in general
position. We use the following lemma.

\begin{lemma} \label{lem:perturb}
  Let $V' = \{\bv'_e\}_{e \in \U'}$ be a collection of vectors in
  $\R^d$. For any $\delta,\gamma>0$, there exists a collection $V'' =
  \{\bv''_e\}_{e \in \U'}$ of vectors in $\R^d$, computable in randomized
  polynomial time in the bit complexity of $V'$, $\log(1/\delta)$, and \(\log\log(1/\gamma)\)
  such that, with probability at least \(1-\gamma\), the vectors $V''$ are in general
  position, and that for all $S \subseteq \U'$ of size $|S| = d$,
\begin{equation}
  \left|
    \det\left(\sum_{e \in S}\bv'_e (\bv'_e)^\top\right)
    -
    \det\left(\sum_{e \in S}\bv''_e (\bv''_e)^\top\right)
  \right| \le \delta.
  \label{eq:lem-det-V-dif}
\end{equation}
\end{lemma}
\begin{proof}
Let \(\sigma>0\) be a constant, to be determined shortly. For each
vector \(\bv_e'\), we add a Gaussian vector with mean $0$ and
covariance matrix \(\sigma \cdot I_d\) to obtain \(\bv_e''\) which has
mean \(\bv_e'\) and covariance \(\sigma \cdot I_d\). For any
\(\sigma>0\), any subset of \(d\) vectors of \(\{\bv''_e\}_{e \in
  \U'}\) are linearly independent with probability 1 (because the set of
singular matrices has Lebesgue measure 0 on \(\R^{d\times d}\)
\cite{caron2005zero}, and the multivariate Gaussian distribution is
absolutely continuous with respect to Lebesgue measure).
Therefore, by the union bound, all \(\binom nd\) such subsets are simultaneously  linearly independent with probability 1, proving that \(V''\) are in general position.

We now show \eqref{eq:lem-det-V-dif}. Let $S \subseteq \U'$ with  $|S|
= d$. Let \(\bV'_S,\bV''_S\) be \(d\times d\) matrices whose columns
are \(\bv_e',\bv_e''\) for \(e\in S\), respectively.  Let
\(\bW=\bV_S''-\bV_S'\). The following inequality is an easy
consequence of the Brunn-Minkowski inequality \cite{ball1997elementary}:
  for any $d\times d$ matrices $\mathbf{A}$ and $\mathbf{\mathbf{B}}$,
  \[
  |\det(\mathbf{A}+\mathbf{\mathbf{B}})|^{1/d} \ge |\det(\mathbf{A})|^{1/d} + |\det(\mathbf{B})|^{1/d}.
  \]
We set \(\mathbf{A}=\bV_S'\) and \(\mathbf{B}=\bW\) to obtain
\[
\det(\bV_S'')^{1/d} \geq \det(\bV_S')^{1/d}+|\det(\bW)|^{1/d}
\]
and set \(\mathbf{A}=\bV''_S\) and \(\mathbf{B}=-\bW\) to obtain
\[
\det(\bV_S')\geq\det(\bV_S'')^{1/d}+|\det(-\bW)|^{1/d}.
\]
Therefore, we have
\begin{equation} \label{eq:V-S-W}
\abs{\det(\bV_S'')^{1/d} - \det(\bV_S')^{1/d}} \leq|\det(\bW)|^{1/d}
\end{equation}

We now bound the determinant of \(\bW\)
with high probability.
Note that \(\bW\) is a random matrix whose entries \(\set{w_{ij}}_{i,j\in[d]}\) are independently sampled from the Gaussian distribution with mean zero and variance \(\sigma\). From a standard tail bound of a Gaussian, we have that for each \(i,j\), 
\[
\Pr[|w_{ij}|>t] \leq\sqrt{\frac 2\pi}\cdot \frac{\sigma e^{-t^2/2\sigma^2}}{t} \leq \frac{\sigma e^{-t^2/2\sigma^2}}{t}.
\]
By union bound, all entries \(w_{ij}\) satisfy \(|w_{ij}|\leq t\) with probability at least \(1-\frac{d^2\sigma e^{-t^2/2\sigma^2}}{t}\).

In this event of probability \(1-\frac{d^2\sigma e^{-t^2/2\sigma^2}}{t}\) , we have from the Leibniz formula for determinants that\begin{align}
|\det(\bW)|&=\abs{\sum_{\tau\in\text{perm}(S)} \text{sgn}(\tau)\prod_{e\in S} w_{i,\tau(i)}}\leq d!\cdot t^d \leq (d t)^d. \label{eq:det-W}
\end{align}
Also, using the convexity of the function \(h(x)=x^{2d}\), we have that for all \(a,b\geq0\),
\[
|a^{2d}-b^{2d}|\leq \abs{a-b}\cdot  [h'(x)]_{x=\max\set{a,b}} = 2d \abs{a-b} \max\set{a,b}^{2d-1}.
\]
Setting \(a=\det(\bV_S')^{1/d}\), \(b=\det(\bV_S'')^{1/d}\) and using \eqref{eq:V-S-W}, we obtain
\begin{align}
\left|
    \det\left(\sum_{e \in S}\bv'_e (\bv'_e)^\top\right)
    -
    \det\left(\sum_{e \in S}\bv''_e (\bv''_e)^\top\right)
  \right| &= \abs{\pr{\det(\bV_S')^{1/d}}^{2d}-\pr{\det(\bV_S'')^{1/d}}^{2d}} \nonumber \\
  &\leq 2d \abs{\det(\bW)^{1/d}} \pr{\det(\bV_S')^{1/d}+|\det(\bW)|^{1/d}}^{2d-1} \nonumber \\
&\leq 2d(dt) \pr{\det(\bV_S')^{1/d}+dt}^{2d-1}, \label{eq:det-V-dif}
\end{align}
where the first equality is by \(\det\left(\sum_{e \in S}\bv'_e (\bv'_e)^\top\right)=\det\pr{\bV_S'\bV_S'^\top}=\det\pr{\bV'_S}^2\) and similarly \(\det\left(\sum_{e \in S}\bv''_e(\bv''_e)^\top\right)=\det\pr{\bV''_S}^2\), and the last equality is by \eqref{eq:det-W}. 

Let \(L\) be the bit complexity of \(V'\). Then we have \(\det\pr{\bV'_S}\leq2^{2L}\) \cite{Schrijver1998}.
Hence, we set \(t=d^{-2}\delta^{-1}{2^{-8L}}\) so that \eqref{eq:det-V-dif} implies
\[
\left|
    \det\left(\sum_{e \in S}\bv'_e (\bv'_e)^\top\right)
    -
    \det\left(\sum_{e \in S}\bv''_e (\bv''_e)^\top\right)
  \right| \leq\delta
\]
as desired.

Recall that this desired bound happens with probability
\(1-\frac{d^2\sigma e^{-t^2/2\sigma^2}}{t}\). Set \(\sigma=\frac{t}{d^2\sqrt{2\log(1/\gamma)} }=\pr{d^4\delta2^{8L+\frac 12}\log^{1/2}(1/\gamma)}^{-1}\) so that
\[
\frac{d^2\sigma e^{-t^2/2\sigma^2}}{t} \leq e^{-\log(1/\gamma)d^2}
\log^{-1/2}(1/\gamma) \leq \gamma
\]
as required. The bit complexity of \(\sigma\) is \(\poly\pr{\log\frac 1\delta,\log\log \frac 1\gamma, L}\). Therefore, the algorithm also runs in time \(\poly\pr{\log\frac 1\delta,\log\log \frac 1\gamma, L}\).
\end{proof}


The next lemma shows that we can, without modifying the optimal value,
replace the original instance by $\M'$ and $V'$.

\begin{lemma}\label{lem:instance-equiv}
  Suppose that
  $\max_{S \in \B} \det\left(\sum_{i \in S}{\bv_i \bv_i^\top}\right) \ne 0$.
  Then for any $\epsilon > 0$, there exists a value of $\delta>0$ such that
  $\log(1/\delta)$ is polynomial in the bit complexity of $V$ and in $\log(1/\epsilon)$, and
  that the vectors $V''$ constructed in Lemma~\ref{lem:perturb} from $V'$ and
  the matroid $\M' = (\U', \I')$ with bases $\B'$ constructed above
  satisfy
  \[
  (1-\epsilon)\max_{S \in \B} \det\left(\sum_{i \in S}{\bv_i \bv_i^\top}\right)
  \le
  \max_{S' \in \B'} \det\left(\sum_{e \in S}{\bv''_e (\bv''_e)^\top}\right)
  \le
  (1+\epsilon)\max_{S \in \B} \det\left(\sum_{i \in S}{\bv_i \bv_i^\top}\right).
  \]
  Moreover, if $S^\star$ achieves the maximum over $\M$ and we let
$S_1 :=
  \{(i,1): i \in S^\star\}$, $S_2 := \{(i,2): i \in S^\star\}$, then
  \[
  \det\left(\frac12\sum_{e \in S_1}{\bv''_e (\bv''_e)^\top} + \frac12 \sum_{e \in S_2}{\bv''_e (\bv''_e)^\top}\right)
  \ge
  (1-\epsilon)\det\left(\sum_{i \in S^\star}{\bv_i \bv_i^\top}\right).
  \]
\end{lemma}
\begin{proof}
  Note that, since $\max_{S \in \B} \det\left(\sum_{i \in S}{\bv_i
      \bv_i^\top}\right) \ne 0$, by Claim~\ref{cl:new-mat} we must have
  \[
  \max_{S \in \B} \det\left(\sum_{i \in S}{\bv_i \bv_i^\top}\right)
  =
  \max_{S \in \B'} \det\left(\sum_{e \in S}{\bv_e (\bv'_e)^\top}\right)
  \ge 2^{-\poly(L)}.
  \]
  where $L$ is the bit complexity of $V$.
  For any $S\subseteq[n]$ of size $k$, by
  Lemma~\ref{lem:perturb} and  the Cauchy-Binet formula, we have
  \[
  \left|
    \det\left(\sum_{e \in S}\bv'_e (\bv'_e)^\top\right)
    -
    \det\left(\sum_{e \in S}\bv''_e (\bv''_e)^\top\right)
  \right| \le \delta{k \choose d}.
  \]
  We can then choose $\delta$ sufficiently small so that
  \[
  \delta{k \choose d}
  \le
  \epsilon   \max_{S \in \B} \det\left(\sum_{i \in S}{\bv_i \bv_i^\top}\right).
  \]
  This choice of $\delta$ suffices for the claim after ``moreover''
  by an analogous reasoning.
\end{proof}

We are now ready to prove Lemma~\ref{lem:alg}, except for the
polynomial time solvability of \CP{}, which is deferred to the next
section. The following lemma captures the non-algorithmic statements
in Lemma~\ref{lem:alg}, when the original instance $(\M, V)$ is
replaced by $(\M', V'')$.

\begin{lemma} \label{lem:relax} Let $\M'$ be as constructed above, and
  let $V''$ be as in Lemma~\ref{lem:instance-equiv}. Let $\Z(\M') = \{
  \bz \in \R^{\U'} : \forall S \in \I_d(\M'), z(S) \geq 0\}$. Then there
  exists a ${\bx^\star}\in \P(\M')\cap \left[0,\frac12\right]^{\U'}$
  such that
  \begin{align*}
  \inf_{\bz\in\Z(\M')}\log\det\pr{\sum_{f \in \U'} x_fe^{z_f} \bv''_f  (\bv''_f)^\top}
  &\ge
  \log\max_{S \in \B} \det\left(\sum_{i \in S}{\bv_i \bv_i^\top}\right)
  - O(\epsilon)\\
  &\ge
  \log\max_{S \in \B'} \det\left(\sum_{f \in S}{\bv''_f (\bv''_f)^\top}\right)
  - O(\epsilon).
  \end{align*}
  Moreover, there exists $\bz^\star$ attaining the infimum on the
  left hand side.
\end{lemma}
\begin{proof}
  Let \(S^\star\in \B\) achieve the maximum in the middle expression. Let \(\bx^\star\in\R^{\U'}\) be defined by \(x_{(i,j)}=\frac12\) if \(i\in S^\star\) and 0 otherwise. To see that \(x^\star\in\P(\M')\), observe that \(S_1:=\set{(i,1):i\in S^\star}\) and \(S_2:=\set{(i,2):i\in S^\star}\) are both bases in \(\M'\), and that \(\bx^\star\) is a convex combination of indicator vectors of those two sets.

For each \(z\in\Z\), we have
\begin{align*}
  \det\pr{\sum_{f\in\U'} \hat x_fe^{z_f} \bv''_f (\bv''_f)^\top}
  &= \det\pr{\frac12\sum_{f\in{S_1}} e^{z_f} \bv''_f (\bv''_f)^\top +
    \frac12\sum_{f\in{S_2}} e^{z_f} \bv''_f (\bv''_f)^\top}\\
  &= \sum_{R \subseteq S_1 \cup S_2: |R| = d} e^{z(R)}2^{-d}
  \pr{\sum_{f \in R}\bv''_f (\bv''_f)^\top},
\end{align*}
where we use the Cauchy-Binet formula for the last equality. For each
\(R\subseteq S_1 \cup S_2\) of size \(d\), we have $R \in \I_d(\M')$,
and therefore $e^{z(R)} \ge 1$, so the right-hand side above is at
least
\begin{align*}
\sum_{R \subseteq S_1 \cup S_2: |R| = d} 2^{-d} \pr{\sum_{f \in  R}\bv''_f (\bv''_f)^\top}
&=
\det\pr{\frac12\sum_{f\in{S_1}} e^{z_f} \bv''_f (\bv''_f)^\top +
    \frac12\sum_{f\in{S_2}} e^{z_f} \bv''_f (\bv''_f)^\top}\\
&\ge (1-\epsilon) \det\pr{\sum_{i \in S^\star} \bv_i \bv_i^\top},
\end{align*}
where we used the Cauchy-Binet formula for the equality and
Lemma~\ref{lem:instance-equiv} for the inequality. This and
Lemma~\ref{lem:instance-equiv} prove the claim before ``moreover''.

The fact that the infimum is attained at some $\bz^\star$ is
guaranteed by Lemma~\ref{lem:finite} and applied to the vectors $V''$
and to the matroid $\M'$. Indeed, the vectors in $V''$ can
be assumed to be in general position by Lemma~\ref{lem:perturb}, and
$\bx^\star$ satisfies $x_e \le \frac12 < 1$ for all $e \in
\U'$. Therefore, the assumptions of Lemma~\ref{lem:finite} hold, and
the infimum is achieved.
\end{proof}

Lemma~\ref{lem:relax} allows us to replace $\M$ and $V$ with $\M'$ and
$V''$, respectively, for the rest of the paper. In particular, it
allows us to assume that  the input vectors  are in general
position, which will be useful in the next section, when we address
polynomial time solvability of the relaxation \CPD{}.

%% file: app-solve.tex
\subsection{Solvability of the Convex Program} \label{sec:solvability-of-CP}



The next theorem shows that \CPD{} can be solved approximately to any
degree of accuracy. This will complete the proof of Lemma~\ref{lem:alg}.
For conciseness we denote \(\X=\P(\M)\cap \left[0,\frac12\right]^n\), \(\I_d=\I_d(\M)\), and \(\P(\I_d)\) the base polytope of the matroid \(([n],\I_d)\).
\begin{theorem} \label{thm:solvability}
There is an algorithm that, given  input vectors \(\bv_1,\ldots,\bv_n\in\R^d\) and $\epsilon>0$, either shows that the optimal value of \DetMax{} is $0$, or returns a solution $(\tilde{\bx},\tilde{\bz})\in\X\times \Z$ such that
\begin{enumerate}
\item $f(\tilde{\bx})\geq f(\bx^\star) -\epsilon$, and\item $g(\tilde{\bx}, \tilde{\bz})\leq \inf_{{\bz\in\Z}} g(\tilde{\bx}, \bz)+\epsilon.$
\end{enumerate}
The algorithm runs in time polynomial  in the size of the input and \(\log\pr{\frac 1\eps}\).
\end{theorem}



The goal of this section is to prove Theorem \ref{thm:solvability}. First, we provide an algorithm to check whether \(\OPT=0\).
\begin{lemma} \label{lem:m-intersect} There exists a polynomial time
  algorithm that, given input vectors \(\bv_1,\ldots,\bv_n\in\R^d\), correctly
  decides if \(\OPT=0\).
\end{lemma}
\begin{proof}
  Observe that, by the Cauchy-Binet formula, for any set \(S\) of size
  \(k\), \(\det\pr{\sum_{i\in S}\bv_i\bv_i^\top}>0 \) if and only if there
  exists a subset \(R\subseteq S\) of size \(d\) such that
  \(\det\pr{\sum_{i\in R}\bv_i\bv_i^\top}> 0 \), i.e. the set of vectors $\{\bv_i: i
  \in R\}$ is linearly independent. Let $\V$ be the collection of
  subsets $R$ of size $d$ of $[n]$ such that $\{\bv_i: i \in R\}$ is
  linearly independent. We have that \(\OPT>0\) if and only if
  $\I_d(\M)$ and \(\V\) have a non-empty intersection.  This can be
  checked in polynomial time using a matroid intersection algorithm
  (see, for example, Chapter 41.2 of \cite{schrijver2003combinatorial}), since both $\I_d(\M)$
  and $\V$ are the collections of bases of matroids.
\end{proof}

Since it is polynomial-time to check if \(\OPT=0\), we may now assume \(\OPT>0\) for the rest of this section. Let \(L\) denote the bit complexity of the input. Since the size of \OPT{} is bounded by \(\poly(L)\) (\cite{Schrijver1998}), we have \(\OPT\in[2^{-p(L)},2^{p(L)}]\) for some polynomial \(p\). Therefore, by binary search on \OPT, the problem of maximizing \(f(\bx)\) reduces to
the feasibility problem \DFea:
\begin{itemize}
\item \textbf{Input}: \(\beta\in\R\) 
\item \textbf{Output:} `yes' if there exists \( \bx\in\P(\M)\) such that  \(f(\bx)\geq \beta\), `no' otherwise.

\end{itemize}

We will solve \DFea{} by the ellipsoid method. We will first show that the feasible set can be bounded with a small loss of accuracy so that the starting ellipsoid contains the feasible set. Then, we will note some properties of the continuous function \(g(\bx,\bz)=\log\det\pr{\sum_{i=1}^n x_ie^{z_i} \bv_i \bv_i^\top}\),   show an oracle to solve the inner problem \(f(\bx)=\inf_{\bz\in\Z}g(\bx,\bz)\), and finally present an algorithm to \DFea.

\begin{lemma} \label{lem:bound-relax}
For all \(\eps>0\), we have 
\begin{equation}
\sup_{\bx\in\X:\bx\geq\frac 1M} \inf_{\bz\in\Z:-M\leq\bz\leq M} g(\bx,\bz)-\eps\leq\sup_{\bx\in\X} \inf_{\bz\in\Z} g(\bx,\bz) \leq\sup_{\bx\in\X:\bx\geq\frac 1M} \inf_{\bz\in\Z:-M\leq\bz\leq M} g(\bx,\bz) + \eps
\end{equation}
for  some \(M\) with bit complexity \(\log M=\poly(L,\log(\frac 1\eps))\).
\end{lemma}
The key to the proof of Lemma \ref{lem:bound-relax}
is to relate 
\(\inf_{\bz\in\Z} g(\bx,\bz)\) with 
\begin{equation} \label{eq:g-as-sup-alpha}
\sup_{\balpha \in\P(\I_d)} \inf_\bz g(\bx,\bz)-\an{\balpha,\bz}
\end{equation}and then apply
Lemma 3.4 of \cite{anari2017generalization}, which states that the supreme over the infimum \eqref{eq:g-as-sup-alpha} can be well-approximated when the feasible set is made bounded.
We now state several claims, setting up for the proof of Lemma \ref{lem:bound-relax}. 

\begin{claim} \label{claim:switch-z-alpha}
Let \(\D\subseteq\R^n\) be a compact set and \(\D'\subseteq\R^n\) be a convex set. Then,
\begin{equation}
\inf_{\bz\in\D'} \sup_{{\balpha} \in\D} g(\bx,\bz)-\an{\balpha,\bz} = \sup_{{\balpha} \in\D} \inf_{\bz\in\D'} g(\bx,\bz)-\an{\balpha,\bz}.
\end{equation}
\end{claim}
\begin{proof}
The function \(g(\bx,\bz)-\an{\balpha,\bz}\) is continuous and concave in \(\bz\) (see \eqref{eq:g-sum-in-z} in the proof of Lemma \ref{lem:g-convex-concave}) and linear in \(\balpha\). Moreover, the domain of $\balpha$  is compact. The claim now follows from Sion's minimax theorem~\cite{Sion58}.
\end{proof}

The claim above implies, in particular, that 
\begin{equation} \label{eq:inf-alpha-sup-z-I-d}
\inf_{\bz} \sup_{{\balpha} \in\P(\I_d)} g(\bx,\bz)-\an{\balpha,\bz} = \sup_{{\balpha} \in\P(\I_d)} \inf_{\bz} g(\bx,\bz)-\an{\balpha,\bz}.
\end{equation}
We now relate \(\inf_{\bz\in\Z} g(\bx,\bz)\) with \(\sup_{\balpha \in\P(\I_d)} \inf_\bz g(\bx,\bz)-\an{\balpha,\bz}\) by claiming that they are equal.
\begin{claim} \label{claim:dual-g}
We have
\begin{equation} \label{eq:dual-g}
\inf_{\bz\in\Z} g(\bx,\bz) = \sup_{\balpha \in\P(\I_d)} \inf_\bz g(\bx,\bz)-\an{\balpha,\bz}.
\end{equation}\end{claim}
\begin{proof}
Let \(\bz^\star\) achieve the infimum in \(\inf_{\bz\in\Z} g(\bx,\bz)\). Then, for all \(\bal\in\P(\I_d)\), we can write \(\balpha=\sum_{S\in\I_d} \lambda_S 1_S\) for \(\lambda_S\geq0\), and we have
\begin{equation}
\an{\balpha,\bz^\star}
=\sum_{S\in\I_d}\lambda_S\an{1_S,\bz^\star}
=\sum_{S\in\I_d}\lambda_S z^\star(S) \geq 0
\end{equation}
where the inequality is by \(\bz^\star\in\Z\). Hence, \[\sup_{\balpha \in\P(\I_d)} \inf_\bz g(\bx,\bz)-\an{\balpha,\bz} \leq \sup_{\balpha \in\P(\I_d)} g(\bx,\bz^\star)-\an{\balpha,\bz^\star} \leq g(\bx,\bz^\star)=\inf_{\bz\in\Z} g(\bx,\bz) .\]

We now show the other direction of the inequality. By \eqref{eq:inf-alpha-sup-z-I-d}, we switch the order of infimum and supremum in \( \sup_{\balpha \in\P(\I_d)} \inf_\bz g(\bx,\bz)-\an{\balpha,\bz}\) to \(\inf_{\bz} \sup_{{\balpha} \in\P(\I_d)} g(\bx,\bz)-\an{\balpha,\bz}\) and let \(\bz^\star\) achieve the \textit{outer} infimum. Since \(g(\bx,\bz^\star)-\an{\balpha,\bz^\star}\) is linear in \(\balpha\), the supremum \(\sup_{{\balpha} \in\P(\I_d)} g(\bx,\bz^\star)-\an{\balpha,\bz^\star}\) can be attained at \(\bal^\star=1_S\) for some \(S\in\I_d\) such that \(\an{S,z^\star}\leq\an{T,z^\star}\) for all \(T\in\I_d\). Let \(\hat \bz=\bz^\star-\frac{\bz^\star(S)}{d}\). Then, for all \(T\in\I_d\), we have
\(
\hat z(T)=\an{T,z^\star}\geq\an{S,z^\star}=z^\star(S)
\),
and therefore \(\hat \bz\in\Z\). 
The feasible solution \(\hat \bz\) gives
\[
g(\bx,\hat\bz)=\log\det\pr{\sum_{i=1}^n x_ie^{z_i^\star-\frac{z^\star(S)}{d}}\bv_i\bv_i^\top}
=\log\pr{\det\pr{\sum_{i=1}^n x_ie^{z_i^\star}\bv_i\bv_i^\top} \cdot e^{-z^\star(S)}}=g(\bx,\bz^\star)-\an{1_S,\bz^\star}
\]
which is the same as the right-hand side of \eqref{eq:dual-g}. Therefore, we have
\[
\inf_{\bz\in\Z} g(\bx,\bz) \leq \sup_{\balpha \in\P(\I_d)} \inf_\bz g(\bx,\bz)-\an{\balpha,\bz}
\]
finishing the proof of the claim.
\end{proof}

We are now ready to prove Lemma \ref{lem:bound-relax}.

\begin{proofof}{Lemma \ref{lem:bound-relax}}
First, we show that the outer supremum is well-approximated after putting a bound on the feasible set:
\begin{equation} \label{eq:outer-bound}
\sup_{\bx\in\X} \inf_{\bz\in\Z} g(\bx,\bz) \geq\sup_{\bx\in\X:\bx\geq\frac 1{M_1}} \inf_{\bz\in\Z} g(\bx,\bz)-\eps 
\end{equation}
for \(M_1=\poly(L,\frac 1\eps)\).
Let \(\bx^\star \in\sup_{\bx\in\X} \inf_{\bz\in\Z} g(\bx,\bz)\). Then, we scale \(\tilde \bx \leftarrow (1-\frac n{k M_1})\bx^\star\) and add the mass \(\frac{n}{M_1}\) to some coordinates of \(\tilde\bx\) so that \(\frac 1{M_1}\leq \tilde \bx\leq \frac 12\). Then for any \(\bz\in\Z\), we have \(g(\tilde \bx,\bz)\geq g(\bx^\star,\bz)-d\log(1-\frac n{k M_1})\). We can bound the error term \(d\log(1-\frac n{k M_1})\) by \(\eps\) by setting \(M_1=\poly(L,\frac 1\eps)\). 

Next, we show that the inner infimum problem is well-approximated after putting a bound on the feasible set.
Let \(\bx\in\X\) such that \(\bx\geq \frac 1M\). By Claim \ref{claim:dual-g}, we have \(\inf_{\bz\in\Z} g(\bx,\bz) = \sup_{\balpha \in\P(\I_d)} \inf_\bz g(\bx,\bz)-\an{\balpha,\bz}\). Now, observe that 
\begin{equation}
\exp(g(\bx,\log \by))=\det\pr{\sum_{i=1}^n x_iy_i \bv_i\bv_i^\top}=\sum_{R\subseteq[n]:|R|=d}\bx^R \by^R\det\pr{\sum_{i\in R} \bv_i\bv_i^\top} 
\end{equation}
is a polynomial of degree \(d\) in \(\by\). By Lemma 3.4 of \cite{anari2017generalization},
\begin{equation} \label{eq:bound-inf}
\sup_{\balpha \in\P(\I_d)} \inf_\bz g(\bx,\bz)-\an{\balpha,\bz} \geq \sup_{\balpha \in\P(\I_d)} \inf_{\bz:-M_2\leq\bz\leq M_2} g(\bx,\bz)-\an{\balpha,\bz} - \eps
\end{equation}
for some \(M_2=\poly(L,\log \frac 1\eps)\). We now claim that\begin{equation} \label{eq:dual-bounded-g}
\sup_{\balpha \in\P(\I_d)} \inf_{\bz:-M_2\leq\bz\leq M_2} g(\bx,\bz)-\an{\balpha,\bz}  \geq\inf_{\bz\in\Z:-M_3\leq\bz\leq M_3} g(\bx,\bz)  
\end{equation}
for some \(M_3=\poly(L,\log \frac 1\eps)\). 

The proof is similar to the proof of Claim \ref{claim:dual-g} as follow(s). First, we switch the supremum and infimum on the left-hand side of \eqref{eq:dual-bounded-g} using Claim \ref{claim:switch-z-alpha}. Second, we fix a solution \(\bz^\star\) of the \textit{outer} infimum and observe that \(\balpha^\star\) achieving \textit{inner} supremum is in the form of \(1_S\) for some \(S\in\I_d\). Next, we construct \(\hat \bz=\bz^\star-\frac{\bz^\star(S)}{d}\) and argue that \(\hat\bz\) is feasible to the right-hand side of  \eqref{eq:dual-bounded-g} and achieves the same objective \(g(\bx,\hat\bz\)) as the left-hand side of \eqref{eq:dual-bounded-g}. This same argument follows here except that the feasibility constraint on the right-hand side of \eqref{eq:dual-bounded-g} contains \(-M_3\leq\bz\leq M_3\). To remedy, we first set \(M_3=2M_2\). By \(-M_2\leq\bz\leq M_2\), we have that \(2M_2\leq\bz^\star-\frac{\bz^\star(S)}{d}\leq 2M_2\), and therefore \(\hat\bz\) is now feasible for the right-hand side of \eqref{eq:dual-bounded-g}.

Combining \eqref{eq:dual-g}, \eqref{eq:bound-inf}, and \eqref{eq:dual-bounded-g}, we obtain
\begin{equation} \label{eq:inner-bound}
\inf_{\bz\in\Z} g(\bx,\bz)\geq\inf_{\bz\in\Z:-M_3\leq\bz\leq M_3} g(\bx,\bz) -\eps.
\end{equation}
Therefore, the statement to be proved follows from \eqref{eq:outer-bound} and \eqref{eq:inner-bound} with \(M=\max\set{M_1,M_3}\).
\end{proofof}

By Lemma \ref{lem:bound-relax}, we now assume that the feasible region are bounded in \([-M,M]^n\) for some \(M\) of polynomial size, which allows us to set the initial ellipsoid of the ellipsoid algorithm. Consequently, the size of \(\bx,\bz\) during the run of the algorithm are polynomial in input size.

Next, we show Lipschitz property of \(g(\bx,\bz)\) and that its Lipschitz constant is bounded by the complexity of the input.\begin{lemma} \label{lem:g-convex-concave}
Let \(\bv_1,\ldots,\bv_n\in\R^d\). The function \(g(\bx,\bz)=\log\det\pr{\sum_{i=1}^n x_ie^{z_i} \bv_i\bv_i^\top}\) is concave in \(\bx\) and convex in \(\bz\). Furthermore, for \(\bx_0\) and \(\bz_0\) of  size polynomial in the input size \(L\)  such that \(\inf_{\bz\in\Z} g(\bx_0,\bz)\) is finite, there exists a polynomial \(p(L)\) such that \(g(\bx,\bz_0)\) as a function of \(\bx\) is \(2^{p(L)}\)-Lipschitz, and that \(g(\bx_0,\bz)\) as a function of \(\bz\) is \({p(L)}\)-Lipschitz.
\end{lemma}
We note that it is required for the Lipschitz constants of \(g(\bx,\bz)\) in \(\bx\) and in \(\bz\) to be at most \(2^{\poly(L)}\) so that the ellipsoid algorithm to be introduced has small error.

\begin{proof}
The proof of concavity of \(g(\bx,\bz)\)  in \(\bx\) follows from the concavity of log-determinant function  \cite{boyd2004convex}. The Lipschitz property in \(\bx\) follows from the calculation of gradient of log-determinant function in the proof in \cite{boyd2004convex}: for a matrix \(\bX\in\SS^n\) and symmetric matrix \(\bY\) of Frobenius norm \(1\), the function \(h(t)=\log\det(\bX+t\bY)\) (on \(t\) such that \(X+t Y \succeq 0)\) has derivative
\begin{align*}
\frac{d}{dt}h(t)&=\frac{d}{dt}\pr{\log\det(\bX)+\log\det\pr{I+t\bX^{-\frac 12} \bY \bX^{-\frac 12}}} \\
&= \frac{d}{d t}\pr{\log\pr{\prod_{i=1}^n(1+t\cdot \lambda_i\pr{ \bX^{-\frac 12} \bY \bX^{-\frac 12}}}} \\
&=\sum_{i=1}^n \frac {\lambda_i\pr{ \bX^{-\frac 12} \bY \bX^{-\frac 12}}}{1+t\cdot \lambda_i\pr{ \bX^{-\frac 12} \bY \bX^{-\frac 12}}}
\end{align*}
where \(\lambda_i\pr{\bM}\) is the \(i\)th eigenvalues of  matrix \(\bM\).
For any \(\bz=\poly(L)\), we have that the Frobenius norm of \(\bX=\sum_{i=1}^n x_ie^{z_i} \bv_i\bv_i^\top\) is at most \(2^{\poly(L)}\) (due to the exponent in \(z_i\)). Hence, for a symmetric matrix \(\bY\) of Frobenius norm \(1\), the Frobenius norm of \(\bX^{-\frac 12} \bY \bX^{-\frac 12}\) is at most \(2^{\poly(L)}\). Hence, \(\abs{\lambda_i\pr{ \bX^{-\frac 12} \bY \bX^{-\frac 12}}}\leq2^{\poly(L)}\) for each \(i\), and therefore \(\frac{d}{d t}h(t)\leq2^{\poly(L)}\). The \(2^{\poly(L)}\)-Lipschitz property on \(h(\lambda)\) implies  \(2^{\poly(L)}\)-Lipschitz property on \(g(\bx,\bz)\) by the chain rule of derivative. 

Observe that, by the Cauchy-Binet formula, 
\begin{equation}
g(\bx,\bz)=\log\pr{\sum_{R\subseteq[n]:|R|=d}\det\pr{\sum_{i\in R} x_ie^{z_i} \bv_i\bv_i^\top}}=\log\pr{\sum_{R\subseteq[n]:|R|=d}e^{z(R)}c_R} \label{eq:g-sum-in-z}
\end{equation}
where \(c_R=\det\pr{\sum_{i\in R} x_i \bv_i\bv_i^\top}\), showing that \(g(\bx,\bz)\) as a function of \(\bz\) is a log-sum-exponential function. The convexity in \(\bz \)   of log-sum-exponential functions is proven in \cite{boyd2004convex}. 

The Lipschitz property of \(g(\bx,\bz)\) in \(\bz\) can be shown by a direct calculation of gradient. Let \(\bx_0\) be such that \(\inf_{\bz\in\Z} g(\bx_0,\bz)\) is finite. Then, 
\begin{align*}
\frac{\partial}{\partial z_i}g(\bx,\bz) &=\frac{\sum_{ R\subseteq[n]:|R|=d,i\in R}e^{z(R)}c_R}{\sum_{R\subseteq[n]:|R|=d}e^{z(R)}c_R} 
\end{align*}
which implies \(0\leq \frac{\partial}{\partial z_i}g(\bx,\bz) \leq 1\) as \(c_R\geq 0\) for all \(R\).
Hence, \(||\nabla_\bz(\bx,\bz) \leq 1||_2\leq\sqrt{n}\).
\end{proof}

Since \(g(\bx,\bz)\) as a function of \(\bz\) is convex, we obtain an efficient oracle \(\F\) to solving the inner infimum problem using an ellipsoid method:
\begin{itemize}
\item 
\textbf{Input: }\(\bx_0\in\X\), error \(\epsilon\).
\item
\textbf{Output: }\(\bz\in\Z\) such that \(g(\bx_0,\bz)\leq\inf_{\bz\in\Z}g(\bx_0,\bz)+\eps\).
\end{itemize}
Because the sizes of the infima and Lipschitz constants are bounded by \(\poly(L)\), the oracle \(\F\) runs in time \(\poly(L,\log(\frac 1\eps))\). 

We now present an oracle \(\G\) for solving the feasibility problem of \CPD, namely \DFea, using the oracle \(\F\).

\begin{lemma}
There exists an oracle \(\G\) for the following problem:

\begin{itemize}
\item 
\textbf{Input:} target \(\beta\in\R\), error \(\epsilon\).
\item
\textbf{Output:} \(\bar \bx\in\X\) such that \(\inf_{\bz\in\Z} g(\bar \bx,\bz) \geq \beta-\eps\), or a proof that \(\sup_{\bx\in\X}\inf_{\bz\in\Z} g(\bx,\bz)<\beta\).
\end{itemize}
Moreover, the  oracle \(\G\) runs in \(\poly(L,\log(\frac 1\eps))\) time.\end{lemma}

\begin{proof}
We assume without loss of generality that for all \(\bx\in\X\), \(\inf_{\bz\in\Z} g(\bx,\bz)\) is finite (by Lemma \ref{lem:m-intersect} and we can check in polynomial time if \(\inf_{\bz\in\Z} g(\bx,\bz)\) is finite) and that \(f(\bx)\) attains its minimum at a finite \(\bz\) (by Lemma \ref{lem:finite}).

We initiate the oracle \(\G\) as an  ellipsoid algorithm with an ellipsoid containing $\X$. At an iteration \(t\) of the algorithm, we denote $\bx_t$ the center of the ellipsoid. If $\bx_t\notin \X$, then we can efficiently find a violating constraint in \(\X\) due to the structure of \(\X\) to get a separating hyperplane and continue. If \(\bx_t\in \X\), we call the oracle \(\F\) with input $\bx_t$ and $\frac \epsilon{4}$. Let $\F$ return $\bz_t$ such that 
$$ g({\bx_t},{\bz_t})\leq \inf_{z\in\Z} g({\bx_t},\bz) +\frac \epsilon4.$$
If $g(\bx_t,\bz_t)\geq \beta-\frac\epsilon2$, then let \(\G\) return $\bx_t$. In this case, we have $$\inf_{z\in\Z} g({\bx_t},\bz)\geq g({\bx_t},{\bz_t})-\frac\epsilon4\geq \beta-\frac\epsilon2-\frac\epsilon4\geq \beta-\eps$$ as needed. 

Otherwise, if 
$g(\bx_t,\bz_t)< \beta-\frac\epsilon2$,  then return the separating hyperplane $\{\bx: \langle \nabla_\bx g(\bx_t,\bz_t),x-\bx_t\rangle \geq 0\}$ to the ellipsoid algorithm (note that \( \nabla_\bx g(\bx,\bz)\) has a closed-form expression and  can be efficiently calculated). We now claim that the returned hyperplane is valid, i.e. that any point $\bx^\star$ such that $\inf_{z\in\Z}g(\bx^\star,\bz) \geq \beta$ satisfies the constraint as given by the separating hyperplane. Let $\bx^\star$ be such that $\inf_{z\in \Z}g(\bx^\star,\bz) \geq \beta$. Thus $g(\bx^\star,\bz_t)\geq \beta$, and therefore $g(\bx^\star,\bz_t)-g(\bx_t,\bz_t)\geq 0$.  By concavity of \(g(\bx,\bz)\) in \(x\), 
\begin{align*}
0\leq g(\bx^\star,\bz_t)-g(\bx_t,\bz_t)\leq \an{\nabla_\bx g(\bx_t,\bz_t),\bx^\star-\bx_t}
\end{align*}
as claimed.

Finally, if the ellipsoid algorithm \(\G\) ends without returning any point after the \(t\)th iteration, then we have $g(\bx_t,\bz_t)<\beta-\frac\epsilon2$. Moreover, if there exists $\bx^\star\in\X$ such that $g(\bx^\star,\bz_t)\geq \beta$, then $\|\bx_t-\bx^\star\|_2\leq \eta_t$ where \(\eta_t\) denotes the radius of the ellipsoid at iteration \(t\). But then we have 
\begin{align}
\frac\eps2<g(\bx^\star,\bz_t)-g(\bx_t,\bz_t)\leq 2^{p(L)} \|\bx^\star-\bx_t\|_2 \leq 2^{p(L)}\cdot \eta_t  \label{eq:ellipsoid-end}
\end{align}
where $2^{p(L)}$ is the Lipschitz constant (Lemma \ref{lem:g-convex-concave}). We run the algorithm until \( 2^{p(L)}\cdot \eta_t < \frac\epsilon2\), so that we get a contradiction in \eqref{eq:ellipsoid-end}. This is the proof that \(\sup_{\bx\in\X}\inf_{\bz\in\Z} g(\bx,\bz)<\beta\), as claimed.

It remains to run the algorithm until \(2^{p(L)}\cdot \eta_t < \frac\epsilon2\), which is equivalent to \(\frac1{\eta_t}=2^{\poly(L)}\cdot\poly(\frac1\eps)\). Since in an ellipsoid algorithm, the radius \(\eta_t\) shrinks exponentially in \(t\), we only need \(T=O(\log \frac{M}{\eta_t})\) iterations, where \(M\) is the size of the initial ellipsoid. By  Lemma \ref{lem:bound-relax},  we may choose \(M\) such that \(\log M=\poly(L,\log(\frac 1\eps))\).
Therefore, we have \(T=O(\log M +\log \frac 1{\eta_t})= \poly(L,\log(\frac1\eps))\). Each of these \(T\) iterations runs in \(\poly(L,\log(\frac 1\eps))\) time, so the total runtime of the ellipsoid algorithm is \(\poly(L,\log(\frac1\eps))\). 
\end{proof}

The existence of the polynomial-time algorithm for \DFea{} finishes the proof of Theorem \ref{thm:solvability}.

%% file: app-KKT.tex
\section{Optimality Conditions}\label{sec:KKT}

In this section, we prove Lemma~\ref{thm:KKT}.
We use Theorem~\ref{thm:slater} (Strong Duality) with
\[
h(\bz) = g(\bx,\bz) = \log \det\pr{\sum_{i=1}^n x_ie^{z_i} \bv_i \bv_i^\top},
\]
where $\bx$ will be either $\bx^\star$ or $\hat{\bx}$. For any fixed $\bx$,
$g(\bx,\bz)$ is convex in $\bz$ (Lemma \ref{lem:g-convex-concave}). We can
write $\Z = \set{\bz\in\R^{n}:z(S)\geq0 \ \ \forall S\in \I_d(\M)}$
as $\Z = \{\bz\in\R^{n}: \bA\bz \ge 0\}$ for a matrix $\bA$ whose rows are
 indicator vectors of sets in $\I_d(\M)$. Thus, for any $\bx$, the
optimization problem $\inf_{\bz \in \Z}g(\bx,\bz)$ can be written as
$\inf\{h(\bz): -\bA\bz \le 0\}$ where $h(\bz) = g(\bx,\bz)$ is a convex function
with domain $\R^{n}$. As we assumed the infimum is achieved, it must
also be finite. Since $\Z$ contains, for example, the non-negative
orthant of $\R^{n}$, it is not empty, and the assumptions of
Theorems~\ref{thm:slater}~and~\ref{thm:gen-KKT} hold. Note that we
apply the theorems with $\ell = 0$.

Let us first show Lemma~\ref{thm:KKT} before ``moreover''.
Observe, first, that by condition~\ref{cond:matrix} of Lemma~\ref{thm:KKT}, we have that
$g(\bx^\star,\bz^\star) =
g(\hat{\bx}, \bz^\star)$. Therefore, to show that $f(\hat{\bx}) =
f({\bx^\star})$, it is enough to show that $\bz^\star$ achieves an
infimum in
$\inf_{\bz \in \Z}
g(\hat{\bx}, \bz)$. We do so by verifying that
conditions~\ref{cond:gen-CS}~and~\ref{cond:gen-first-order} of
Theorem~\ref{thm:gen-KKT} hold for $\bz^\star$ and $\blambda$. Indeed, condition~\ref{cond:CS} of
Lemma~\ref{thm:KKT} is exactly the condition~\ref{cond:gen-CS} of
Theorem~\ref{thm:gen-KKT}. Moreover, the function $g(\hat{\bx}, \bz)$ is
differentiable in $\bz$, and its partial derivatives at $\bz^\star$ are given
by
\[
\frac{\partial g}{\partial z_i}(\hat{\bx}, \bz^\star) = \hat x_i e^{z_i^\star}\bv_i^\top \bX^{-1} \bv_i,
\]
for $\bX = \sum_{i=1}^n \hat x_i e^{z_i^\star}\bv_i \bv_i^\top$. We also have
$(\mathbf{M}^\top \blambda)_i = -(\bA^\top \blambda)_i = -\sum_{S\in \I_d(\M): i\in S}\lambda_S$, and,
therefore, condition~\ref{cond:first-order} of Lemma~\ref{thm:KKT} is exactly
condition~\ref{cond:gen-first-order} of Theorem~\ref{thm:gen-KKT}. Lemma~\ref{thm:KKT} before ``moreover'' now follows by Theorem~\ref{thm:gen-KKT}.

Next we establish Lemma~\ref{thm:KKT}  after ``moreover''. The Lagrangian of
$\inf_{\bz \in \Z}{g(\bx^\star, \bz)}$ equals
\[
L(\bz, \lambda) = g(\bx^\star,\bz) - \sum_{S \in\I_d(M)} \lambda_S z(S).
\]
Theorem~\ref{thm:slater} implies that there exists $\blambda^\star
\in \R^{\I_d(\M)}_{\geq 0}$ such that
\[
\inf_{\bz \in \R^{n}} L(\bz,\blambda^\star) = \inf_{\bz \in \Z}{g(\bx^\star, \bz)} = g(\bx^\star, \bz^\star).
\]
Therefore, conditions~\ref{cond:gen-CS}~and~\ref{cond:gen-first-order}
hold for $\bz^\star$ and $\blambda^\star$, and, as we argued above, they are
equivalent to conditions~\ref{cond:CS}~and~\ref{cond:first-order} of
Lemma~\ref{thm:KKT}.



%% file: app-sparsity.tex
\section{Proofs from Section~\ref{sec:sparse}} \label{sec:proof-sparse}
For a vector \(\bx\in\R^{n}\) over the ground set \([n]\) and a smaller ground set $\tilde \U\subseteq[n]$, we denote by \(\bx_{|\tilde\U}\in\R^{\tilde \U}\) the vector \(\bx\) restricted to \(\tilde \U\).
We show that restricting the ground set to the support of a solution of \(\CPD\) preserves the optimal value and optimal solutions (after restricting to the support\ of the inner infimum. 
\begin{lemma} \label{lem:restrict-U}
Let \(\bx^\star\) be a feasible solution to \CPD. Let \(\tilde\U=\supp{\bx}\) and let
 \(\tilde \bx^\star=\bx^\star_{|\tilde\U}\). Let \(\tilde \M=(\tilde \U,\tilde \I)\) be a matroid where
\[
\tilde \I= \set{S\subseteq \tilde \U: S\in\I}.
\]
 Define \(\tilde\Z\) and \(\tilde g:\R^{\tilde \U}\times \R^{\tilde \U}\rightarrow \R\) as \(\Z\) and \(g\) restricted to \(\tilde \U\) by
\[
\tilde \Z=\set{\tilde\bz\in\R^{\tilde \U}:\tilde z(S)\geq0 \ \ \forall S\in \I_d(\tilde \M)}
\]
and
\[
\tilde g(\tilde\bx,\tilde\bz)=\log\det\pr{\sum_{i\in\tilde\U}x_ie^{z_i} \bv_i\bv_i^\top}.
\]
Then,
\begin{equation}
\inf_{\bz\in\Z} g(\bx^\star,\bz) = \inf_{\tilde\bz\in\tilde\Z} \tilde g(\tilde \bx^\star,\tilde\bz).
\end{equation}
Moreover, for any \(\bz^\star\in\inf_{\bz\in\Z} g(\bx^\star,\bz) \), we have \(\hat \bz^\star_{|\tilde\U}\in\inf_{\tilde\bz\in\tilde\Z} \tilde g(\tilde \bx^\star,\tilde\bz)\).
\end{lemma}
\begin{proof}
We first show that \(\inf_{\bz\in\Z} g(\bx^\star,\bz) \geq \inf_{\tilde\bz\in\tilde\Z} \tilde g(\tilde \bx^\star,\tilde\bz)\). For any \(\bz\in\Z\), 
observe that \(\bz_{|\tilde\U}\in\tilde\Z\) and that by the definitions of \(g,\tilde g\) we have \(g(\bx^\star,\bz)=\tilde g(\tilde \bx^\star,\bz_{|\tilde\U})\). Hence, \(\inf_{\bz\in\Z} g(\bx^\star,\bz) \geq \inf_{\tilde\bz\in\tilde\Z} \tilde g(\tilde \bx^\star,\tilde\bz)\). 

To prove the other direction of the inequality, let \(\tilde \bz \in\tilde\Z\). We then construct \(\bz\in\Z\) from \(\tilde \bz\) by adding \(z_j\) for each \(j\in[n]\setminus\U'\) with value \(z_j\geq \sum_{i\in\U'}|z_i|\). This ensures that \(z(S)\geq0\) for all \(S\in\I_d(\M)\), and hence \(\bz\in\Z\). Again, \(g(\bx^\star,\bz)=\tilde(\tilde \bx^\star,\tilde \bz)\), so we have \(\inf_{\bz\in\Z} g(\bx^\star,\bz) \leq \inf_{\tilde\bz\in\tilde\Z} \tilde g(\tilde \bx^\star,\tilde\bz)\).
Therefore,  \(\inf_{\bz\in\Z} g(\bx^\star,\bz) = \inf_{\tilde\bz\in\tilde\Z} \tilde g(\tilde \bx^\star,\tilde\bz)\), as claimed.

Next, let \(\bz^\star\in\inf_{\bz\in\Z} g(\bx^\star,\bz) \). Since \(g(\bx^\star,\bz^\star)=\tilde g(\tilde \bx^\star,\bz^\star_{|\tilde\U})\) and \(\inf_{\bz\in\Z} g(\bx^\star,\bz) = \inf_{\tilde\bz\in\tilde\Z} \tilde g(\tilde \bx^\star,\tilde\bz)\), we have \(\hat \bz^\star_{|\tilde\U}\in\inf_{\tilde\bz\in\tilde\Z} \tilde g(\tilde \bx^\star,\tilde\bz)\).
\end{proof}

The following statement that minimum-weight bases form a matroid is standard; we include its proof for completeness.
\begin{lemma}\label{lem:bases}
Let $\M=([n], \I)$ be a matroid and let $\B$ denote the set of bases of $\M$. Let $\bw:[n]\rightarrow \R$  denote a weight function and let $\B'=\{S\in \B: w(S)=\min_{T\in \B} w(T)\}$ denote the set of minimum weight bases of \(\M\). Then,\, $\B'$ are bases of another matroid $\M'=([n],\I')$. Moreover, if $\M$ admits an independence oracle, then $\M'$ admits an independence oracle. 
\end{lemma}
\begin{proof}
A set system \(\B'\) are bases of a matroid if it has an exchange property:
\(
\forall S_1,S_2\in \B'\) such that \(S_1\neq S_2, \
\)
we have that for all \( v\in S_2\setminus S_1\), there is \(u\in S_1\setminus S_2\) such that \(S_1+v-u\in \B'\) (\cite{whitney1935abstract}). 

Let \(S_1,S_2\in\B'\) be such that \(S_1\neq S_2\). Let \(v\in S_2\setminus S_1\). Since \(S_1,S_2\) are bases of  matroid \(\M\), by the strong basis exchange property, there exists \( u\in S_1\setminus S_2 \) such that \(S_1+v-u,S_2+u-v\in\B\).
We consider different cases based on \(w_u,w_v\).

If \(w_u<w_v\), then \(w\pr{S_2+u-v} < w(S_2)\), a contradiction to \(S_2\) being a minimum-weight basis.
If \(w_u>w_v\), then \(w\pr{S_1+v-u} > w(S_1)\), again a contradiction to \(S_1\) being a minimum-weight basis.
Therefore, \(w_u=w_v\), and so \(S_1+v-u\) is also a minimum weight basis, as desired.

To test if a set $Q$ is an independent set in $\M'$, we check if $Q$ is an independent set in $\M$ and if the minimum-weight independent set in $\M/Q$ is equal to $\min_{T \in \B} w(T) - w(Q)$. Here, $\M/Q$ denote the matroid $\M$ after contracting $Q$. Since we can check if a set is independent in $\M$ in polynomial time and optimize a linear function over a matroid constraint in polynomial time, we can check if $Q$ is independent in $\M'$ in polynomial time.
\end{proof}

We now complete a missing proof of Lemma \ref{lem:extreme-LP-uncross} using the
uncrossing technique.

\begin{proofof}{Lemma \ref{lem:extreme-LP-uncross}}
Let \(\bx^\star\) be an extreme solution to \LP. A chain \(\C_1\) corresponding to tight linearly independent constraints in \eqref{mat11}-\eqref{mat12} of a matroid base polytope can be obtained by an uncrossing argument (see Lemma 5.2.4 of \cite{lau2011iterative}). We will show that a chain \(\C_2\) with similar property can be obtained for  constraints \eqref{mat21}-\eqref{mat22}. Let \(\ba=(e^{z_i^\star} \bv_i \bX^{-1} \bv_i)_{i\in[n]}\). For \(S\subseteq[n]\), we denote \(\ba_{S}\) a vector obtained from \(\ba\) by setting \(a_i=0\) for each coordinate \(i\notin S\). Let \(\F=\set{\emptyset\subsetneq S\subseteq [n]:\sum_{i\in S} x_i  e^{z_i^\star} \bv_i \bX^{-1} \bv_i = r^\star(S)}\) be the set of tight constraints in \eqref{mat21}-\eqref{mat22}.
The uncrossing argument applies to show that \(\F\) is also closed under union and intersection as follow(s).
\begin{lemma}
If \(A,B\in\F\), then \(A\cup B,A\cap B\in \F\). Moreover, \(\ba_A+\ba_B=\ba_{A\cup B} + \ba_{A\cap B}\).
\end{lemma}
\begin{proof}
The proof follows similarly from the proof for a base polytope (see Lemma 5.2.2 of \cite{lau2011iterative}).
We have
\begin{align*}
r^\star(A)+r^\star(B)&=\sum_{i\in A} x_i  a_i+\sum_{i\in B} x_i  a_i \\
 &=\sum_{i\in A\cup B} x_i  a_i+\sum_{i\in A\cap B} x_i  a_i \\
&\leq r^\star(A\cup B)+r^\star(A\cap B)\\
&\leq r^\star(A)+r^\star(B)
\end{align*}
The first equality is by \(A,B\in\F\). The first inequality follows from constraints \eqref{mat21}. The last inequality follows from submodularity of rank function of a matroid. The equality \(\ba_A+\ba_B=\ba_{A\cup B} + \ba_{A\cap B}\) is straight-forward from the basic set property.
\end{proof}

The rest of the proof to show an existence of \(\C_2\) follows similarly from the standard uncrossing argument (Lemma 5.2.4 of \cite{lau2011iterative}). Note that chains \(\C_1,\C_2\) obtained from the
uncrossing argument are in the same linear program with some tight constraints \(P\) in \eqref{vec1}. However, we may remove linearly dependent constraints when we take the set of constraints in \(\C_1,\C_2,P\) together until we have linearly independent constraints.

Finally, \LP{}, which has \(n\) variables,  must have \(n\) linearly independent constraints to specify an extreme solution. Since the number of tight constraints in \eqref{mat11}-\eqref{vec1} is \(|\C_1|+|\C_2|+|P|\), there are \(n-|\C_1|+|\C_2|+|P|\) tight constraints in \eqref{LP:x-geq-0}. Therefore, \(|\supp{\bx}|=|\C_1|+|\C_2|+|P|\).
\end{proofof}

%% file: app-round.tex
\section{Proofs from Section \ref{sec:rounding} } \label{sec:app-rounding}

\begin{proofof}{Lemma \ref{lem:relate}} We prove this lemma using the inequality proven for log-concave polynomials in~\cite{AnariGV18} (see Lemma~\ref{lem:AGV_bound} in Appendix) by setting up an appropriate polynomial.
For a given $\bx^\star$, let $g(y_1,\dots,y_n)  = \det\pr{\sum_{i=1}^n y_i x_i^\star \bv_i \bv_i^\top}$. For a matroid $\M = ([n], \I)$, let
\[ W = \{ [n]\setminus S \mid S \in \I_d\}.\]
Let $\I^\star $ be the set of all subsets of sets in $W$. That is, 
\begin{equation} \label{eq:dual-matroid}
\I^\star= \{ S \subseteq [n] \mid \exists T \in W \text{ such that }S \subseteq T\}.
\end{equation}
The following claim follows from the fact that for any matroid, independent sets of a fixed size form a basis of another matroid and that complements of these independent sets form a basis of the dual matroid (see Chapter 2, Theorem 1 in \cite{Welsh10}).
\begin{claim}
For a matroid \(\M=([n],\I)\) and \(\I^\star\) defined as in \eqref{eq:dual-matroid}, the set system $\M^\star=([n], \I^\star)$ is a matroid with basis set $W= \{ [n]\setminus S \mid S \in \I_d\}$.
\end{claim}

Let $h(z_1,\dots,z_n)$ be the bases generating polynomial of $\M^\star$. That is,
\[ h(z_1,\dots,z_n) = \sum_{S \in \I_d}\Pi_{i \in [n]\setminus S}  z_i = \sum_{S \in \I_d} z^{[n]\setminus S} \]
We use several lemmas, which can be found in Appendix \ref{sec:prelim}. For a nonzero scalar \(a\in\R\) and a vector \(\bp\in\R^n\), we let \(\frac \bp a := [\frac{p_i}{a}]_{i=1}^n\) be a vector obtained from element-wise division. For vectors \(\by,\bz,\bp\in\R^n\), we let \(\by^\bp:=\prod_{i=1}^ny_i^{p_i}\) and \(\bz^{1-\bp}:=\prod_{i=1}^n z_i^{1-p_i}\). By Lemmas \ref{lem:matroid_log_concave} and \ref{lem:det-is-log-concave}, both $g$ and $h$ are completely log-concave polynomials. By Lemma~\ref{lem:prod_complete_log_concave}, $g(\by)h(\bz)$ is a completely log-concave polynomial. Hence, by Lemma~\ref{lem:AGV_bound}, for any $\mathbf{p} \in [0,1]^n$,

\begin{equation}\label{eq:anari_gv_identity}
 \pr{\Pi_{i=1}^n ({\partial}_{y_i} + {\partial}_{z_i})} g(\mathbf{y})h(\mathbf{z}) |_{\mathbf{y} = \mathbf{z} = 0} \geq \pr{\frac{\mathbf{p}}{e^2}}^{\mathbf{p}} \inf_{\by,\bz \in \rR_{>0}^n} \frac{g(\by)h(\bz)}{\by^{\mathbf{p}} \bz^{1-\mathbf{p}}}.
\end{equation}

We first simplify the left-hand side of \eqref{eq:anari_gv_identity}.
\begin{claim}
We have $\pr{\Pi_{i=1}^n ({\partial}_{y_i} + {\partial}_{z_i})} g(\mathbf{y})h(\mathbf{z}) |_{\mathbf{y} = \mathbf{z} = 0} = \sum_{S \in \I_d} \det\pr{\sum_{i\in S} x_i^\star \bv_i \bv_i^\top}$
\end{claim}
\begin{proof}
By the Cauchy-Binet formula,
\[ g(\by) = \det\pr{\sum_{i=1}^n y_i x_i^\star \bv_i \bv_i^\top} = \sum_{S \in {[n]\choose d}} \det\pr{\sum_{i\in S} y_i x_i^\star \bv_i \bv_i^\top} = \sum_{S \in {[n]\choose d}} y^S\det\pr{\sum_{i\in S}  x_i^\star \bv_i \bv_i^\top}.  \]
By definition, $h(\bz) = \sum_{S \in \I_d} z^{[n] \setminus S}$. Hence,
\[ g(\by)h(\bz) = \sum_{S_1 \in {[n]\choose d}} \sum_{S_2 \in \I_d} y^{S_1} z^{[n] \setminus S_2} \det\pr{\sum_{i\in S_1}  x_i^\star \bv_i \bv_i^\top}.\]
So, we have
\begin{align*}
\pr{\Pi_{i=1}^n ({\partial}_{y_i} + {\partial}_{z_i})} g(\mathbf{y})h(\mathbf{z})   & = \pr{\sum_{T \subseteq [n]} \Pi_{i \in T} {\partial}_{y_i} \Pi_{j \in [n] \setminus T} {\partial}_{z_j} } \sum_{S_1 \in {[n]\choose d}} \sum_{S_2 \in \I_d} y^{S_1} z^{[n] \setminus S_2} \det\pr{\sum_{i\in S_1}  x_i^\star \bv_i \bv_i^\top}\\
& =  \sum_{S_1 \in {[n]\choose d}} \sum_{S_2 \in \I_d} \pr{\sum_{T \subseteq [n]} \pr{\Pi_{i \in T} {\partial}_{y_i} \Pi_{j \in [n] \setminus T} {\partial}_{z_j}}y^{S_1} z^{[n] \setminus S_2} }  \det\pr{\sum_{i\in S_1}  x_i^\star \bv_i \bv_i^\top}.
\end{align*}

It is easy to see that for $T \subseteq [n]$ and $S_1, S_2$ such that $|S_1| = |S_2|$, $\pr{\Pi_{i \in T} {\partial}_i \Pi_{j \in [n] \setminus T} {\partial}_j}y^{S_1} z^{[n] \setminus S_2}$ is equal to $1$ if $T = S_1 = S_2$ and $0$ otherwise. Hence,
\[\pr{\Pi_{i=1}^n ({\partial}_{y_i} + {\partial}_{z_i})} g(\mathbf{y})h(\mathbf{z}) |_{\mathbf{y} = \mathbf{z} = 0} = \sum_{S \in \I_d} \det\pr{\sum_{i\in S} x_i^\star \bv_i \bv_i^\top}\]
finishing the proof of the claim.
\end{proof}

Next, we reformulate the right-hand side of \eqref{eq:anari_gv_identity}. For vectors \(\by,\bw,\bp\in\R^n\), we let \((\by\bw)^\bp:=\prod_{i=1}^n (y_iw_i)^{p_i}\) and \(\bw^{\bp-1}:=\prod_{i=1}^n w_i^{p_i-1}\).

\begin{claim}
We have
\[\inf_{\by,\bz \in \rR_{>0}^n} \frac{g(\by)h(\bz)}{\by^{\mathbf{p}} \bz^{1-\mathbf{p}}} = \inf_{\by, \bw \in \rR_{>0}^n} \frac{\det\pr{\sum_{i=1}^n x_i^\star y_i \bv_i \bv_i^\top} \pr{\sum_{S \in \I_d} w^S }}{\pr{\by\bw}^{\bp}}
\]
\end{claim}
\begin{proof}
By a change of variable \(\bw=1/\bz\) coordinate-wise, we have\[ \inf_{\by,\bz \in \rR^n_{>0}} \frac{g(\by) h(\bz)}{\by^{\bp} \bz^{1-\bp}} = \inf_{\by, \bw \in \rR_{>0}^n} \frac{g(\by)h(1/\bw)}{\by^{\bp} \bw^{\bp-1}}\]
Substituting $g$ and $h$ by their definitions, we get
\begin{align*}
\inf_{\by,\bz \in \rR^n_{>0}} \frac{g(\by) h(\bz)}{\by^{\bp} \bz^{1-\bp}}& =  \inf_{\by, \bw \in \rR_{>0}^n} \frac{\det\pr{\sum_{i=1}^n x_i^\star y_i \bv_i \bv_i^\top} \pr{\sum_{S \in \I_d} w^S/w^{[n]}}}{\by^\bp \bw^{\bp-1}}\\
&=\inf_{\by, \bw \in \rR_{>0}^n} \frac{\det\pr{\sum_{i=1}^n x_i^\star y_i \bv_i \bv_i^\top} \pr{\sum_{S \in \I_d} w^S }}{\pr{\by\bw}^{\bp}}
\end{align*}
as claimed.
\end{proof}

Applying the two claims above to the left- and right-hand sides of \eqref{eq:anari_gv_identity}, we get that for any $\bp \in [0,1]^n$,
\[ \sum_{S \in \I_d} \det\pr{\sum_{i\in S} x_i^\star \bv_i \bv_i^\top} \geq \pr{\frac{\mathbf{p}}{e^2}}^{\mathbf{p}}\inf_{\by, \bw \in \rR_{>0}^n} \frac{\det\pr{\sum_{i=1}^n x_i^\star y_i \bv_i \bv_i^\top} \pr{\sum_{S \in \I_d} w^S }}{\pr{\by\bw}^{\bp}}.\]
In particular, if we consider all $\bp \in \P(\I_d)$, we get
\[\sum_{S \in \I_d} \det\pr{\sum_{i\in S} x_i^\star \bv_i \bv_i^\top} \geq \sup_{\bp \in \P(\I_d)} \pr{\frac{\mathbf{p}}{e^2}}^{\mathbf{p}}\inf_{\by, \bw \in \rR_{>0}^n} \frac{\det\pr{\sum_{i=1}^n x_i^\star y_i \bv_i \bv_i^\top} \pr{\sum_{S \in \I_d} w^S }}{\pr{\by\bw}^{\bp}}.\]
For any $\bp \in \P(\I_d)$, we have $\sum_{i=1}^n p_i = d$. Hence,
\[\sum_{S \in \I_d} \det\pr{\sum_{i\in S} x_i^\star \bv_i \bv_i^\top} \geq e^{-2d} \sup_{\bp \in \P(\I_d)} \inf_{\by, \bw \in \rR_{>0}^n} \frac{\det\pr{\sum_{i=1}^n x_i^\star y_i \bv_i \bv_i^\top} \pr{\sum_{S \in \I_d} w^S }}{\pr{\frac{\by\bw}{\bp}}^{\bp}}.\]
By changing variable $\bp$ to ${\balpha}$, we get the desired result.
\end{proofof}

\begin{proofof}{Lemma~\ref{lem:stronger_cp}}
Let $R = \inf_{\bz\in \Z} \det\pr{\sum_{i=1}^n x_i^\star e^{z_i} \bv_i \bv_i^\top}$.
By the change of variable $y_i = e^{z_i}$, we get 
\begin{equation}\label{eq:temp_prelim_1}
R = \inf_{\by > 0: \forall S \in \I_d, y^S \geq 1} \det\pr{\sum_{i=1}^n x_i^\star y_i \bv_i \bv_i^\top}.
\end{equation}
We now claim a condition to check the feasibility of \(\by\) of the infimum \eqref{eq:temp_prelim_1}.
\begin{claim}
For any $\by \in \rR_{\geq 0}^n$,
\begin{center}
$y^S \geq 1$ for all $S \in \I_d$ if and only if $y^{\balpha} \geq 1$ for all ${\balpha} \in \P(\I_d)$.
\end{center}
\end{claim}
\begin{proof}
Suppose $y^S \geq 1$ for all $S \in \I_d$. Let ${\balpha} \in \P(\I_d)$. Then, there exists $\blambda \in \rR^{\I_d}_{\geq 0}$ such that $\sum_{S \in \I_d} \lambda_S = 1$ and ${\balpha} = \sum_{S \in \I_d} \lambda_S 1_S$ where $1_S$ is an indicator vector of a set $S$.
Then, we have
\[ y^{\balpha} = y^{\sum_{S \in \I_d} \lambda_S 1_S} = \Pi_{S \in \I_d} \pr{\Pi_{i \in S} y_i}^{\lambda_S} \geq \Pi_{S \in \I_d} 1^{\lambda_S} \geq 1,\]
proving one direction of the claim.
Next, suppose that $y^{\balpha} \geq 1$ for any ${\balpha} \in \P(\I_d)$. Note that $1_S \in \P(\I_d)$ for any $S \in \I_d$, so we may use \(y^{\balpha} \geq 1\) with \(\bal=1_S\). Hence, $y^S \geq 1$ for any $S \in \I_d$.
\end{proof}

Applying the above claim to \eqref{eq:temp_prelim_1}, we get
\[ R = \inf_{\by > 0: \forall {\balpha} \in \P(\I_d), y^{\balpha} \geq 1} \det\pr{\sum_{i=1}^n x_i^\star y_i \bv_i \bv_i^\top}.\]
Since $\det\pr{\sum_{i=1}^n x_i^\star y_i \bv_i \bv_i^\top}$ is a degree $d$ polynomial in $\by$,
\begin{align*}
R &= \inf_{\by > 0} \frac{\det\pr{\sum_{i=1}^n x_i^\star y_i \bv_i \bv_i^\top}}{\inf_{{\balpha} \in \P(\I_d)} \by^{\balpha}}  =  \inf_{\by > 0} \sup_{{\balpha} \in \P(\I_d)} \frac{\det\pr{\sum_{i=1}^n x_i^\star y_i \bv_i \bv_i^\top}}{\by^{\balpha}}.
\end{align*}
Applying Claim \ref{claim:switch-z-alpha} (Sion's minimax theorem) on \(\log R\), we get
\begin{align}
R=\inf_{\by >0} \sup_{{\balpha} \in \P(\I_d)} \frac{\det\pr{\sum_{i=1}^n x_i^\star y_i \bv_i \bv_i^\top}}{\by^{\balpha}} = \sup_{{\balpha} \in \P(\I_d)} \inf_{\by > 0} \frac{\det\pr{\sum_{i=1}^n x_i^\star y_i \bv_i \bv_i^\top}}{\by^{\balpha}}. \label{eq:R-y-no-w}
\end{align}

Next, we relate the right-hand side of \eqref{eq:R-y-no-w} to  the left-hand side of the inequality in Lemma~\ref{lem:stronger_cp} by the following claim. We denote \(\pr{\frac{\bw}{{\balpha}}}^{{\balpha}}:=\prod_{i=1}^n \pr{\frac {w_i} {\alpha_i}}^{\alpha_i} \).
\begin{claim}\label{claim:max_entropy}
For any $\bw \geq 0$ and ${\balpha} \in \P(\I_d)$, we have $\sum_{S \in \I_d} w^S \geq \pr{\frac{\bw}{{\balpha}}}^{{\balpha}}$.
\end{claim}
\begin{proof} We assume that ${\balpha}$ is strictly inside the base polytope with the base set $\I_d$. If not, we can focus on the matroid with bases corresponding to the vertices of the smallest face in $\P(\I_d)$ containing ${\balpha}$. By Proposition 2.3 in~\cite{FeichtnerS05}, every face of a matroid polytope is a matroid polytope.

Setting $\zeta = \I_d$ and $\mathbf{p} = {\balpha}$ in Lemma~\ref{lem:max_entropy_dist}, we get a distribution $\mu: 2^{\I_d} \rightarrow \rR_+$ and $\lambda_1,\dots,\lambda_n > 0$ such that $\mu(S) \propto \lambda^S$ for $S \in \I_d$. Moreover, ${\balpha}_i = \Pr_{S \sim \mu} [ i \in S]$ and ${\balpha} = \sum_{S \in \I_d} \mu(S) 1_S$. The generating polynomial for $\mu$ is $g_{\mu}(\bz) = \frac{1}{\sum_{S \in \I_d} \lambda^S} \sum_{S \in \I_d} \lambda^S z^S$, which we claim to be log-concave. By Lemma~\ref{lem:matroid_log_concave}, $\sum_{S \in \I_d} z^S$ is log-concave. By Lemma~\ref{lem:log_concavity_preserve}, substituting $z_i$ by $\lambda_i z_i$ in and multiplying with a constant $\frac{1}{\sum_{S \in \I_d} \lambda^S}$ to the polynomial \(\sum_{S \in \I_d} z^S\) preserve log-concavity. Hence, $g_{\mu}(\bz)$ is log-concave, as claimed, and so $\mu$ is a log-concave distribution. 

By Lemma~\ref{lem:compare_entropy},
\[ \sum_{S \in \I_d } \mu(S) \log \frac{1}{\mu(S)} \geq \sum_{i =1}^n {\balpha}_i \log \frac{1}{{\balpha}_i}\]
which is equivalent to
\[ \prod_{S \in \I_d} \frac{1}{\pr{\mu(S)}^{\mu(S)}} \geq \prod_{i=1}^n \frac{1}{{\balpha}_i^{{\balpha}_i}} =\frac{1}{{\balpha}^{\balpha}}.\]

Now, we are ready to prove the claim. We have
\begin{align*}
\pr{\frac{\bw}{{\balpha}}}^{\balpha} & = \frac{\bw^{\balpha}}{{\balpha}^{\balpha}} = \frac{\bw^{\pr{\sum_{S \in \I_d} \mu(S) 1_S}}}{{\balpha}^{\balpha}}\\
& \leq \frac{\Pi_{S \in \I_d} \bw^{\mu(S)1_S}}{\Pi_{S \in \I_d} \mu(S)^{\mu(S)}} = \Pi_{S \in \I_d} \pr{\frac{\bw^{1_S}}{\mu(S)}}^{\mu(S)}\\
& \leq \sum_{S \in \I_d} \bw^{1_S} = \sum_{S \in \I_d} w^S.
\end{align*}
where the last inequality follows from the weighted AM-GM inequality since $\sum_{S \in \I_d} \mu(S) = 1$.
\end{proof}

We continue of the proof of the lemma. By \eqref{eq:R-y-no-w} and Claim~\ref{claim:max_entropy}, for any $\bw \geq 0$, we have
\[ R \leq  \sup_{{\balpha} \in \P(\I_d)} \inf_{\by \geq 0} \frac{\det\pr{\sum_{i=1}^n x_i^\star y_i \bv_i \bv_i^\top}}{y^{\balpha}} \frac{\sum_{S \in \I_d} w^S}{\pr{\frac{\bw}{{\balpha}}}^{{\balpha}}}.\]
Therefore,
\[R=\inf_{\bz\in \Z} \det\pr{\sum_{i=1}^n x_i^\star e^{z_i} \bv_i \bv_i^\top} \leq \sup_{{\balpha} \in \P(\I_d)} \inf_{\by,\bw>0} \frac{\det\pr{\sum_{i=1}^n x_i^\star y_i \bv_i \bv_i^\top} \pr{\sum_{S \in \I_d} w^S }}{\pr{\frac{yw}{{\balpha}}}^{\balpha}} \]
finishing the proof of the lemma.
\end{proofof}

%% file: no-pc-dist.tex
\section{Oblivious Rounding Scheme}\label{sec:no_pc_dist}

In this section, we show that none of the previous approaches for \DetMax\ yield an approximation factor independent of the size of the output solution $k$ even if the dimension of the vectors $d$ is $2$. Formally, we show that any relaxation and rounding schemes satisfying the following properties cannot achieve an approximation factor independent of $k$. 
\begin{itemize}
\item Let the relaxation be $\sup_{\bx \in \P(\M)} g(\bx)$ for some function $g$. Then, for any $\bx  \in \P(\M)$ which we write as \(\bx= \sum_{T \in \B} \lambda_T 1_T\) for \( \sum_{T \in \B} \lambda_T =1\) and \(\lambda\geq0\), we have $g(\bx) \geq \max_{T \in \B} \det\pr{\lambda_T \sum_{i \in T} \bv_i \bv_i^\top}$.
\item Given $\bx \in \P(\M)$, the rounding scheme outputs a solution $T \in \B$ with probability dependent only on $\bx$ and $\M$ (and so independent of $\bv_i$'s). 
\end{itemize}

We construct an   instance as follow(s).
\paragraph{Matroid $\M$ and $\bx \in \P(\M)$:} Consider the graphic matroid with a graph $G$ on $n+2$ vertices $V = \{a_1,\dots,a_m\} \cup \{b,c\}$ and an edge set $E = \{a_ib,a_ic \mid i \in [m]\}$. All  spanning trees of $G$ are  bases of  matroid \(\M\). Consider a fractional spanning tree $\bx^\star$ such that $x_e^\star = \frac{m+1}{2m}$ for every edge $e \in E$. 

Let the rounding scheme pick a subset of edges with distribution $\mu :2^{E} \rightarrow \rR_+$. Since the rounding scheme outputs a basis of the matroid, it must be that $\mu(F) = 0$ if the graph $(V,F)$ has a cycle. Suppose we sample a subgraph as per distribution $\mu$. We let
\begin{center}
$B_i:=$ the event that both $a_ib$ and $a_ic$ are picked.
\end{center}
We now prove some properties about these events.
\begin{claim}For any $i \neq j$, $\Pr[B_i \cap B_j] = 0$. Hence, $\exists i \in [m]$ such that $\Pr[B_i] \leq \frac{1}{m}$.
\end{claim}
\begin{proof}
If both $B_i$ and $B_j$ occur, then our sampled subgraph contains edges $a_ib,a_ic,a_jb,a_jc$ which implies that there is a cycle in the subgraph. However, by the definition, $\mu(F) = 0$ if the sampled subgraph $(V,F)$ has a cycle. Hence, $\Pr[B_i \cap B_j] =0$. Therefore, $\sum_{i=1}^m \Pr[B_i] =\Pr[\cup_{i=1}^m B_i] 
 \leq 1$, so there exists  $i \in [m]$ such that $\Pr[B_i] \leq \frac{1}{m}$. 
\end{proof}

We continue constructing the instance with the description of input vectors.
\paragraph{Vector Set:} Consider the vector set as follows: $\bv_{a_ib} = \left[\begin{array}{c} 2\\ 0\end{array}\right], \bv_{a_ic} = \left[\begin{array}{c} 0 \\ 2 \end{array}\right]$, and for $j \neq i$, $\bv_{a_jb} = \bv_{a_jc} = \left[ \begin{array}{c} 0 \\ 0 \end{array} \right]$. Since the rounding scheme is oblivious to the set of vectors, we can make such a selection.

By the assumption on the relaxation, we have that $g(\bx^\star) \geq \det\pr{\frac{m+1}{2m}\bv_{a_ib} \bv_{a_ib}^\top + \frac{m+1}{2m}\bv_{a_ic}\bv_{a_ic}^\top} = \frac{(m+1)^2}{m^2} >1$. The rounding scheme with distribution $\mu$ outputs a solution with non-zero value only if both $a_i u$ and $a_iv$ are picked. Hence, the expected objective value of the solution returned is 
\[\Pr[B_i] \det\pr{\bv_{a_ib}\bv_{a_ib}^\top + \bv_{a_ic}\bv_{a_ic}^\top} \leq \frac{4}{m}\] 
and the approximation factor achieved is larger than $\frac{1}{\frac{4}{m}} = \frac{m}{4}$. As $m \rightarrow \infty$, the approximation factor tends to infinity even for $d=2$. 

To construct a similar instance for $d>2$, we  add vectors $\bv_i = \left[\begin{array}{c} 0^{i-1}\\ 1 \\ 0^{d-i}\end{array}\right]$ for  each $i \in \{3,\dots,d\}$ and include them in  the bases of the matroid.

%% file: partition-approx.tex
\section{Improved Approximation for a Partition Matroid}\label{sec:partition_approx}

In this section, we show an $e^{3d}$-estimation algorithm for \DetMax{} under a partition matroid.
Algorithm \ref{alg:deterministic} and the same analysis of the algorithm will  imply an efficient derandomization  with approximation factor \(\exp(O(d^3))\) for a partition matroid.
\begin{theorem}\label{thm:partition}
  There is an efficiently computable convex program whose objective value estimates the objective of \DetMax{} problem under a partition matroid constraint within a multiplicative factor of $e^{O(d)}$.
\end{theorem}
We start by discussing the rounding scheme  presented in Algorithm~\ref{alg:partition_rounding}.
\begin{algorithm}
\begin{algorithmic}[1]
\State \para{Input:} a partition matroid $\M=([n],\I)$ with bases $\B$, and $\bx \in \P(\M)$\State \para{Output:} a basis $T \in \B$
\State Sample a set $T \in \B$ with probability $\frac{x^T}{\sum_{R \in \B} x^R}$
\State Return $T$\end{algorithmic}
\caption{Rounding Scheme for a Partition Matroid}
\label{alg:partition_rounding}
\end{algorithm}

To see that Algorithm~\ref{alg:partition_rounding} is polynomial time, observe that if we sample a set $W$ of $b_i$ elements from the partition $P_i$ with probability
proportional to $x^W$ for each $i$, then our sample would be a set $T \in \B$ with probability $\frac{x^T}{\sum_{R \in B} x^R}$. Such a sampling can be done efficiently as proved by Singh and Xie~\cite{SinghX18}. Next, we show that for every independent set $S$ of size $d$, we sample a basis containing $S$ with a large probability.
\begin{lemma}\label{lem:partition_prob}
Let $T$ denote the random set returned by Algorithm~\ref{alg:partition_rounding}. Then, for any set $S \in \I_d$, we have
\[ \Pr[S \subseteq T] \geq e^{-d} x^S.\]
\end{lemma}

The statement then implies a lower bound on the expected objective value of the solution returned.

\begin{lemma}\label{lem:partition_expected}
Algorithm~\ref{alg:partition_rounding} returns a basis $T \in \B$ with expected objective value
\[ \E\left[ \det\pr{\sum_{i \in T} \bv_i \bv_i^\top} \right] \geq e^{-d} \sum_{S \in \I_d} \det \pr{\sum_{i \in S} x_i \bv_i \bv_i^\top}\]
\end{lemma}
Next, we relate this lower bound of the objective of the convex relaxation \CPD{} by using Lemma~\ref{lem:relate} and Lemma~\ref{lem:stronger_cp}. Before we prove these lemmas, we prove Theorem~\ref{thm:partition}.

\begin{proofof}{Theorem~\ref{thm:partition}} We start by solving the convex relaxation \CPD{} (which can be done in polynomial time from Theorem \ref{thm:solvability}). Let $\bx^\star$ be an optimal solution to \CPD{} (same argument works for a near optimal solution as well). Let $T \in \B$ be the random solution returned by Algorithm~\ref{alg:partition_rounding} with input $\bx = \bx^\star$. By Lemma~\ref{lem:partition_expected}, the expected value of the solution returned is
\[\E\left[ \det\pr{\sum_{i \in T} \bv_i \bv_i^\top} \right] \geq e^{-d} \sum_{S \in \I_d} \det \pr{\sum_{i \in S} x_i \bv_i \bv_i^\top}.\]
By Lemmas~\ref{lem:relate} and~\ref{lem:stronger_cp}, the right-hand side of the above inequality is further bounded, and  we get \[\E\left[ \det\pr{\sum_{i \in T} \bv_i \bv_i^\top} \right] \geq e^{-d} \cdot e^{-2d}  \inf_{\bz \in \Z} \det\pr{\sum_{i=1}^n x_i^\star e^{z_i} \bv_i \bv_i^\top}.\]

Since $\bx^\star$ is an optimal solution to \CPD{}, we have $\inf_{\bz \in \Z} \det\pr{\sum_{i=1}^n x_i^\star e^{z_i} \bv_i \bv_i^\top} = \CP{}$ which is at least $\OPT$. Hence, we get a random solution $T$ in polynomial time with expected value
\[\E\left[ \det\pr{\sum_{i \in T} \bv_i \bv_i^\top} \right] \geq e^{-3d} \cdot \OPT{}\]
which finishes the proof.
\end{proofof}

To prove Lemma~\ref{lem:partition_prob}, we make use of a similar result proved by Singh and Xie~\cite{SinghX18} in the context of a uniform matroid.

\begin{theorem}\label{thm:singhx_uniform}(Proposition 2 in~\cite{SinghX18}) For a uniform matroid with rank at least $d$ and a fractional solution $\by$ in the matroid polytope, if we sample a basis $Q$ with probability $y^Q$, then for each set $W$ of size $d$, all elements of $W$ are selected with probability at least $e^{-d} y^W$.

More formally, let $\by \in \rR_+^m$ be a vector such that $\sum_{i=1}^m y_i = \ell$ where $\ell$ is an integer. Then for any $W \subseteq [m]$ such that $|W| \leq \ell$, we have
\[ \frac{\sum_{Q \in {[m]\choose \ell}: W \subseteq Q} y^Q}{\sum_{Q \in {[m] \choose \ell}} y^Q} \geq e^{-|W|} y^W.\]
\end{theorem}

\begin{proofof}{Lemma~\ref{lem:partition_prob}} A set $T \in \B$ is sampled with probability $\frac{x^T}{\sum_{R \in \B} x^R}$. Hence, for any set $S \in \I_d$,
we have
\begin{equation} \label{eq:prob-S-in-T-partition}
\Pr[S \subseteq T] = \sum_{T \in \B: S \subseteq T} \frac{x^T}{\sum_{R \in \B} x^R} =  \frac{\sum_{T \in \B: S\subseteq T} x^{T}}{\sum_{R \in \B} x^R}.
\end{equation}
Let the partition matroid be $\M = ([n],\I)$ with partitions $P_1,\dots,P_t$ such that $\cup_{i=1}^t P_i = [n]$ and let the rank of $P_i$ be $b_i$. For any $i \in [t]$, let $ S_i = S \cap P_i $. Then, the numerator and denominator of the right-hand side of \eqref{eq:prob-S-in-T-partition} can be decomposed into products across each partition as $\sum_{R \in \B} x^R = \prod_{i=1}^t \pr{\sum_{R \in {P_i \choose b_i}} x^{R}}$ and $\sum_{T \in \B: S\subseteq T} x^{T}  = \prod_{i=1}^t \pr{\sum_{T\in {P_i \choose b_i}: S_i \subseteq T} x^{T}}$.
Therefore,
\[ \Pr[S \subseteq T] =   \prod_{i=1}^t \frac{\sum_{T\in {P_i \choose b_i}: S_i \subseteq T} x^{T}}{\sum_{R \in {P_i \choose b_i}} x^{R}}.
\]
Since $\bx \in \P(\M)$, we have $\sum_{i \in P_j} x_i = b_j$. Since $S \in \I_d$, we also have $|S_j| = |S \cap P_j| \leq b_j = \sum_{i \in P_j} x_i $. Applying Theorem~\ref{thm:singhx_uniform}, we get

\[\Pr[S \subseteq T] \geq \prod_{i=1}^t e^{-|S_i|} x^{S_i} = e^{-\sum_{i=1}^t |S_i|} x^{\cup_{i=1}^t S_i}.
\]
Since $S \in \I_d$, we have $d=|S| = |\cup_{i=1}^t S_i| = \sum_{i=1}^t |S_i|$. Therefore, \[\Pr[S \subseteq T]\geq e^{-d} x^S\]
as desired.
\end{proofof}

We now prove Lemma~\ref{lem:partition_expected}.

\begin{proofof}{Lemma~\ref{lem:partition_expected}} Let $T \in \B$ be the random set returned by Algorithm~\ref{alg:partition_rounding}. Then, by the Cauchy-Binet formula,
we have\[ \E\left[ \det\pr{\sum_{i \in T} \bv_i\bv_i^\top}\right] = \E\left[ \sum_{S \subseteq T: |S| = d} \det\pr{\sum_{i \in S} \bv_i \bv_i^\top}\right] = \sum_{S \in {[n] \choose d}} \Pr[S \subseteq T] \det\pr{\sum_{i \in S} \bv_i \bv_i^\top}.\]
By Lemma~\ref{lem:partition_prob}, for each set $S \in \I_d$, we have
$\Pr[S \subseteq T] \geq e^{-d} x^S$. Therefore,
\[ \E\left[ \det\pr{\sum_{i \in T} \bv_i\bv_i^\top}\right] \geq \sum_{S \in \I_d} e^{-d} x^S\det\pr{\sum_{i \in S} \bv_i \bv_i^\top} = e^{-d} \sum_{S \in \I_d} \det\pr{\sum_{i \in S} x_i \bv_i \bv_i^\top}.\]

\end{proofof}